\documentclass[11pt,a4paper,reqno]{amsart}
\usepackage[hmargin={2.3cm,2.3cm},vmargin={2.5cm,2.5cm},centering]{geometry}
\usepackage{amssymb,amsmath,amsthm,amsfonts,amscd}

\usepackage{graphicx}
\usepackage[dvipsnames]{xcolor}
\usepackage{mathrsfs}
\usepackage{hyperref}
\usepackage{tikz-cd}
\usepackage[utf8]{inputenc}
\usepackage[english]{babel}
\usepackage{dsfont}
\usepackage{bm,bbm}
\usepackage[longnamesfirst,numbers,sort&compress]{natbib}

\usepackage{enumerate}
\usepackage{setspace}
\usepackage[font={scriptsize,bf}]{caption}
\usepackage{changepage}
\usepackage{booktabs}
\usepackage[sort&compress,capitalise,nameinlink]{cleveref}
\crefname{section}{\textsection}{\textsection}
\crefname{subsection}{\textsection}{\textsection}
\crefname{subsubsection}{\textsection}{\textsection}
\crefname{paragraph}{\textparagraph}{\textparagraph}
\crefname{thm}{Theorem}{Theorem}

\newtheorem{theorem}{Theorem}[section]

\newtheorem{proposition}[theorem]{Proposition}
\newtheorem{assumption}[theorem]{Asumption}
\newtheorem{corollary}[theorem]{Corollary}

\theoremstyle{definition}
\newtheorem{definition}[theorem]{Definition}

\theoremstyle{remark}

\renewcommand\phi{\varphi}

\newcommand{\cH}{\mathcal{H}}

\newcommand{\g}{\gamma}

\renewcommand{\epsilon}{\varepsilon}

\renewcommand{\ge}{\geqslant}
\renewcommand{\le}{\leqslant}
\renewcommand{\geq}{\geqslant}
\renewcommand{\leq}{\leqslant}

\renewcommand{\tilde}{\widetilde}

\numberwithin{equation}{section}

\begin{document}

\title{On some rigorous aspects of fragmented condensation}

\author[D.\ Dimonte]{Daniele Dimonte}

\address{Universität Basel\\ Spiegelga{\ss}e 5\\4051 Basel}

\email{daniele.dimonte@unibas.ch}

\urladdr{daniele.dimonte.it}

\author[M.\ Falconi]{Marco Falconi}

\address{Mathematics Department\\Università di Roma Tre\\Largo San Leonardo Murialdo
  1/C\\00146 Roma}

\email{mfalconi@mat.uniroma3.it}

\urladdr{http://ricerca.mat.uniroma3.it/users/mfalconi/}

\author[A.\ Olgiati]{Alessandro Olgiati}

\address{ CNRS Grenoble\\ LPMMC\\ F-38000 Grenoble}

\email{alessandro.olgiati@lpmmc.cnrs.fr}



\date{October 11, 2019}

\thanks{The final and revised version of this manuscript, dated July 2020,
  will appear in
  \href{https://iopscience.iop.org/journal/0951-7715}{Nonlinearity}. In
  compliance with the journal's rules, it will be made available on arXiv 12
  months after publication.}

\begin{abstract}
  In this paper we discuss some aspects of fragmented condensation from a mathematical perspective. We
  first propose a simple way of characterizing finite fragmentation. Then, inspired by recent results
  of semiclassical analysis applied to bosonic systems with infinitely many degrees of freedom, we address
  the problem of persistence of fragmented condensation. We show that the latter occurs in interacting
  systems, in the mean-field regime, and in the limit of large gap of the one-body Hamiltonian.
\end{abstract}
\maketitle


\section{Introduction}
\label{sec:introduction}

The phenomenon of fragmented Bose-Einstein condensation (fragmented BEC) has attracted a lot of attention
in recent years, both from an experimental and a theoretical point of view
\cite{Mueller2006,leggett2008oup,bader2009prl,fischer2010pra,kang2014prl}. The physical idea of
fragmentation is that in some cases the condensed fraction of a bosonic system is distributed among
multiple, finitely or infinitely many, one single-particle states. This is in contrast with the occurrence
of ordinary, or simple, Bose-Einstein condensation in which a single one-particle orbital is
macroscopically occupied. In this paper we discuss two aspects of fragmented condensation: its
mathematical definition and its persistence under time evolution.

An intuitive, and very common (see, \emph{e.g.}, \cite{leggett2008oup}), definition of finite fragmentation is
the following one. Consider a system of $N$ bosons in a quantum state with density matrix $\gamma_N$,
normalized so that $\mathrm{Tr}\gamma_N=1$, and assume that the associated reduced one-particle density matrix
(1-RDM) $\gamma^{(1)}_N$ converges, as $N\to \infty$ (in a suitable topology), to a limit one-particle density matrix
$\gamma^{(1)}_{\infty}$. It is then customary to say that $\gamma_N$ exhibits fragmented condensation if $\gamma_\infty^{(1)}$ has
\emph{a finite rank strictly larger than one}.

Despite its simplicity, the above definition is too broad, namely, it includes states that cannot be
physically interpreted as fragmented. We will exhibit in Section \ref{sec:char-fragm-cond} a class of
states that are statistical mixtures of ordinary Bose-Einstein condensates but whose effective
one-particle reduced density matrices have rank two or more. Hence, the above definition also includes
states in which, with probability one, the system is observed as an ordinary BEC. In order to exclude at
least these cases, it is therefore desirable to provide a different definition/characterization of finite
fragmented condensation.

We propose in Definition \ref{def:fragm_nostra} a simple characterization of finite fragmentation that
distinguishes between fragmented condensation and statistical mixtures of simply condensed states: we will
say that a state exhibits fragmented condensation if the rank of $\gamma^{(1)}_{\infty}$ is two or more, and the
rank of the effective $p$-particle reduced density matrix ($p$-RDM) $\gamma^{(p)}_{\infty}$ is a non-constant
function of $p$. We justify in Section \ref{sec:char-fragm-cond} this characterization by comparing the
behavior of the effective $p$-particle reduced density matrix of zero-temperature fragmented states with
the ones of statistical mixtures of ordinary condensates.


Another interesting aspect of finite fragmented condensation is its persistence in time under the effect
of interactions. It is well-known that ordinary BEC is preserved in time if the pair interaction among
particles is proportional to the inverse of the number of particles (mean-field, or Hartree,
regime)~\cite{ginibre1979cmpI,ginibre1979cmp2,rodnianski2009cmp,knowles2010cmp,Benedikter-Porta-Schlein-2015}. This
means very weak interactions, but occurring on macroscopic scales. In such a regime
the effective equation ruling the evolution of the condensate wave-function is the non-local Hartree
equation. Simple BEC is known to be preserved even in the physically more relevant ultra-dilute
(Gross-Pitaevskii)
regime~\cite{ErdSchYau-10,Pickl-10,pickl2015rmp,benedikter2015cpam,brennecke2017arxiv}. In this case the
scaling prescription assigns a coupling constant to the pair interaction proportional to $N^2$, while the
interaction range shrinks as $N^{-1}$. Interactions happen therefore on short scales, and they are very
strong. The effective equation ruling the evolution of the condensate wave-function is the celebrated
Gross-Pitaevskii equation, \emph{i.e.}, a local cubic non-linear Schr\"odinger equation that takes into
account inter-particle correlations.

Treating fragmented condensation with mean-field techniques is, in turn, more complicated. This is due to
the fact that fragmented states are intrinsically more correlated. A set of effective evolution equations
were formally derived in \cite{alon2007plA,alon2008prA} by imposing that, in the large $N$ limit, finite
fragmentation persists at any time, with a fixed number of one-particle states. This idea seems close to
the so-called Dirac-Frenkel principle used in mathematical physics~\cite{Benedikter2018} and numerical
analysis~\cite{lubich2008zlam}. Unfortunately, the error that one commits when imposing that fragmentation
with the same number of states holds at any time is in general not converging to zero as $N\to \infty$. This can
be deduced as a consequence of a series of rigorous results on the mean-field effective evolution of
generic many-body states in the Hartree regime
\cite{ammari2008ahp,ammari2009jmp,ammari2011jmpa,ammari2015asns,ammari2016cms}. In fact, it follows from
these results that finite fragmentation is in general destroyed by an interaction in the mean field
regime: for almost all times the time-evolved effective one-particle reduced density matrix
$\gamma^{(1)}_{\infty}(t)$ has \emph{infinite rank}, even if the initial datum has $\gamma^{(1)}_\infty(0)$ finite rank
(strictly larger than one).

As a first step towards addressing these difficulties, we prove in Theorem \ref{thm:convergence} that
finite fragmented condensation is preserved in interacting systems in the mean-field regime which present
a very large energy gap between the degenerate ground states and the first excited states of the one-body
Hamiltonian; this requires the introduction of a further parameter $\omega$ in the theory (the first one being
the number of particles $N$). Having a large gap prevents the spreading of the wavefunction on the whole
Hilbert space, effectively constraining the system to a finite dimensional subspace, and thus preserving
finite fragmentation. We will show that there is a well-defined limit theory
describing the system as $\omega\to\infty$ and $N\to\infty$, and we will explicitly exhibit the effective equations for the
one-particle states on which the system condenses.

Natural follow-up problems would be the study of fragmented condensation at finite inverse temperature in
the thermodynamic limit, and the proof of persistence of fragmented condensation in the dilute regime and
with large gap. Both those problems would however require different mathematical techniques from the ones
used in this paper. We plan to address them in future works.

The rest of the paper is organized as follows. In Section \ref{sec:char-fragm-cond} we propose and justify
our definition of fragmented BEC. In Section \ref{sec:pers-finite-fragm} we present the precise setting
and assumptions for our result on persistence of fragmented BEC in the large gap limit, and state our main
result, Theorem \ref{thm:convergence}. Section \ref{sec:proof-crefthm:convergence1} and Section
\ref{sec:proof-crefthm:convergence2} contain the proof of Theorem \ref{thm:convergence}. Section
\ref{sec:proof-crefthm:matrices} contains the proof of Proposition \ref{prop:matrices} in which we compute
the effective $p$-RDM of particular fragmented states. Finally, in Appendix \ref{sec:mean-field-descr} we
outline the applications of semiclassical techniques to many-body bosonic systems that are most relevant
to our analysis.

\section{Characterization of fragmented condensation}
\label{sec:char-fragm-cond}

In this Section we propose and justify a characterization of finite fragmented condensation which, while
differing slightly from the one commonly adopted, has the advantage of excluding statistical mixtures of
ordinary BECs.

\subsection{Definition of finite fragmented BEC}
Let us consider a bosonic system of $N$ particles with one-body Hilbert space $\mathcal{H}$ and many-body
Hilbert space $\mathcal{H}_N=\mathcal{H}^{\otimes_\mathrm{sym} N}$. For any $p\in\mathbb{N}$, we will
denote by $\mathcal{H}_p$ the $p$-particle Hilbert space $\mathcal{H}^{\otimes_\mathrm{sym} p}$. Let
$\gamma_N$ be a quantum state, \emph{i.e.}, a positive, trace-class, operator on $\mathcal{H}_N$ with unit trace.
For any $p=1,\dots,N$, we define the $p$-particle reduced density matrix ($p$-RDM) $\gamma_N^{(p)}$
associated with $\gamma_N$ to be the positive, trace-class operator on $\mathcal{H}_p$ which satisfies
\begin{equation*}
  \mathrm{Tr}_{\mathcal{H}_p} \big( A\gamma^{(p)}_N\big)=\mathrm{Tr}_{\mathcal{H}_N}\big(A\gamma_N\big)
\end{equation*}
for any bounded operator $A$ on $\mathcal{H}_p$.

The physical notion of Bose-Einstein condensation is that of macroscopic occupation of one-body
orbitals. To make this mathematically more precise we introduce the limit, or effective, reduced
marginals. Even though, for any fixed $p$, the sequence of $\gamma_N^{(p)}$ do not in general converge as
$N\to\infty$, it is however always possible to extract a subsequence such that there exists the limit
\begin{equation} \label{eq:limit_marginal}
  \gamma_\infty^{(p)}:=\lim_{j\to\infty} \gamma^{(p)}_{N_j}
\end{equation}
in the weak-* sense. By a diagonal procedure we can actually ensure that convergence occurs for all $p$'s
along the same subsequence. It follows that for each $p$ the operator $\gamma_\infty^{(p)}$ is postive and
trace-class in $\mathcal{H}_p$. In the concrete situations that we consider in this work we will actually
have strong convergence in trace norm for the whole sequence, which ensures that
$\mathrm{Tr}\gamma_\infty^{(p)}=1$. The concept of BEC is then usually mathematically expressed as
follows.

\begin{definition}[\textbf{Finite Bose-Einstein condensation \`a la Penrose-Onsager~
\cite{penrose1956pr,leggett2008oup}}]\mbox{}\\ \label{def:penr_ons}
A bosonic system in the (sequence of) state(s) $\gamma_N$ is said to exhibit finite Bose-Einstein
condensation if for each $p\in\mathbb{N}$ the limit \eqref{eq:limit_marginal} exists as a positive
trace-class operator on $\mathcal{H}_p$ \emph{and}
\begin{itemize}
\item $\gamma_\infty^{(1)}$ has rank one, in which case the BEC is \emph{simple};
\item $\gamma_\infty^{(1)}$ has finite rank strictly larger than one, in which chase the BEC is
  \emph{fragmented}.
\end{itemize}
\end{definition}
It is possible to interpret the hierarchy of effective marginals $\gamma_\infty^{(p)}$ with the help of a measure
$\mu$ on the one-particle Hilbert space $\mathcal{H}$ (see Appendix \ref{sec:mean-field-descr} for more
details, in particular \eqref{eq:5} and the ensuing discussion). More precisely,
there exists a positive measure $\mu$ on $\mathcal{H}$, independent on $p$, such that
\begin{equation} \label{eq:conv_comb} \gamma_{\infty}^{(p)}=\int_{\mathcal{H}}^{}\, \bigl\lvert
  \,\underbrace{\varphi\otimes \dotsm\otimes \varphi}_{p} \,\bigr\rangle\bigl\langle\,
  \underbrace{\varphi\otimes \dotsm\otimes \varphi}_{p} \,\bigr\rvert \,\mathrm{d}\mu(\varphi)\; .
\end{equation}
and $\int_\mathcal{H}d\mu(\varphi)= 1$. According to this interpretation, the measure of a simple condensate is
either concentrated on a single point $\varphi_0\in\mathcal{H}$, or it is a convex combination of measures
concentrated on points differing from each other only by a phase. These are indeed the only measures
yielding
\begin{equation*}
  \gamma_{\infty,\mathrm{simp}}^{(1)}=\lvert \,\varphi_0 \,\rangle\langle\, \varphi_0 \,\rvert\; ,
\end{equation*}
in \eqref{eq:conv_comb}. It follows that, for a simple condensate,
\begin{equation*}
  \gamma_{\infty,\mathrm{simp}}^{(p)}=\bigl\lvert \,\underbrace{\varphi_0\otimes \dotsm\otimes \varphi_0}_{p} \,\bigr\rangle\bigl\langle\, \underbrace{\varphi_0\otimes \dotsm\otimes \varphi_0}_{p} \,\bigr\rvert\;,
\end{equation*}
i.e., all effective marginals have rank one.

Definition \ref{def:penr_ons}, while being perfectly satisfactory in the case of simple BEC, is however
\emph{not sufficient} to characterize fragmented condensates. We propose instead the following definition.

\begin{definition}[\textbf{Finite fragmented Bose-Einstein condensation}]\mbox{}\\ \label{def:fragm_nostra}
  A bosonic system in the (sequence of) state(s) $\gamma_N$ is said to exhibit finite fragmented Bose-Einstein
  condensation if the limit \eqref{eq:limit_marginal} exists for each $p$ as a positive trace-class
  operator on $\mathcal{H}_p$, the rank $R(p)$ of $\gamma^{(p)}_\infty$ is a non-constant finite function of $p$, and
  $R(1)\ge2$.
\end{definition}
We will motivate the validity of Definition \ref{def:fragm_nostra} in the rest of this Section.

\subsection{Effective $p$-RDM's of fragmented states and of mixtures of simple BEC's}

In order to justify Definition \ref{def:fragm_nostra} we will explicitly exhibit states which:
\begin{itemize}
\item Should not be reasonably interpreted as describing a fragmented condensate;
\item Would be admitted by Definition \ref{def:penr_ons};
\item Are \emph{not} admitted by our Definition \ref{def:fragm_nostra}.
\end{itemize}
We will also show that Definition \ref{def:fragm_nostra} is satisfied by a class of states that are
certainly interpreted as fragmented.

\begin{assumption}[\textbf{One-particle states and relative weights}]\mbox{}\\ \label{assum:one_body}
  Let $\{\varphi_k\}_k$ be a set of orthonormal one particle wave functions in $\mathcal{H}$ for $k=1,\dots,M$
  with some fixed $M\in\mathbb{N}$. Consider a sequence of integers $\bigl(f_k(N)\bigr)_{N\in \mathbb{N}}$ such
  that
  \begin{equation} \label{eq:f_pi} \sum_{k=1}^Mf_k(N)=N,\;\;\forall
    N\qquad\text{and}\qquad\exists\,\pi_k:=\lim_{N\to\infty}\frac{f_k(N)}{N}\in[0,1].
  \end{equation}
  Correspondingly, define the set of indices corresponding to non-zero $\pi_k$'s
  \begin{equation} \label{eq:P} P=\Bigl\{k=1,\dotsc,M\;|\; \pi_k\neq 0\Bigr\},
  \end{equation}
  with cardinality
  \begin{equation}
    \label{eq:18}
    \Pi:= \mathrm{Card}(P)\leq M.
  \end{equation}
\end{assumption}

For $\{\varphi_k,f_k(N)\}_{k=1}^M$ as in Assumption \ref{assum:one_body}, define the statistical mixture
of pure condensates
\begin{equation}
  \label{eq:16}
  \gamma_{N,\mathrm{stat}}=\sum_{k=1}^{M}\dfrac{f_k(N)}{N}\,\bigl\lvert\, \underbrace{\varphi_k\otimes \dotsm\otimes \varphi_k}_N\, \bigr\rangle\bigl\langle\,\underbrace{\varphi_k\otimes \dotsm\otimes \varphi_k}_N\,\bigr\rvert.
\end{equation}
This state will be observed with probability $f_k(N)/N$ as the ordinary BEC $\varphi_k^{\otimes N}$, for
$k=1,\dots,M$. It is therefore, to our understanding, \emph{not interpretable as a fragmented condensate}.

We can, on the other hand, define a genuinly fragmented state whose 1-RDM coincide with that of
\eqref{eq:16}. For any $\varphi\in\mathcal{H}_m$ and $\theta\in\mathcal{H}_n$, we denote by $\varphi\vee \theta\in \mathcal{H}_{m+n}$
their symmetric tensor product. Given $\{\varphi_k,f_k(N)\}_{k=1}^M$ as in Assumption \ref{assum:one_body},
define the pure state
\begin{equation} \label{eq:gamma_frag}
  \begin{split}
    \gamma_{N,\mathrm{frag}}=\;& \bigl\lvert\, (\underbrace{\varphi_1\otimes \dotsm\otimes \varphi_1}_{f_1(N)})\vee\dotsc\vee (\underbrace{\varphi_{2s+1}\otimes \dotsm\otimes \varphi_{2s+1}}_{f_{2s+1}(N)})\,\bigr\rangle\\
    &\quad \bigl\langle\, (\underbrace{\varphi_1\otimes \dotsm\otimes \varphi_1}_{f_1(N)})\vee\dotsc\vee (\underbrace{\varphi_{2s+1}\otimes \dotsm\otimes \varphi_{2s+1}}_{f_{2s+1}(N)})\,\bigr\rvert\;.
  \end{split}
\end{equation}
For each $k$, exactly $f_k(N)$ particles in the state $\gamma_{N,\mathrm{frag}}$ will be observed in the
one-body state $\varphi_k$.  We will study and compare the p-RDM of $\gamma_{N,\mathrm{stat}}$ and
$\gamma_{N,\mathrm{frag}}$ in detail below.

It is worth remarking that fragmented states analogous to $\gamma_{N,\mathrm{frag}}$ have been considered
in~\cite{Rougerie2016}, where the authors analyze a system of interacting bosons trapped by a suitably
scaled double-well confining potential. If one denotes by $u_{\mathrm{l}}$ and $u_{\mathrm{r}}$ the ground
state wavefunctions localized in the left and right well respectively, then it is expected for bosons to
form a fragmented condensate of the form $u_{\mathrm{l}}^{\otimes \frac{N}{2}}\vee u_{\mathrm{r}}^{\otimes \frac{N}{2}}$,
provided that the relative distance of the two wells gets sufficiently large as $N\to \infty$. This is to be
compared with the opposite situation, in which the distance does not grow fast enough, and therefore
particles tend to form a simple condensate of the type
$(\frac{u_{\mathrm{l}}+u_{\mathrm{r}}}{\sqrt{2}})^{\otimes N}$. The authors prove that the fragmented behavior
is energetically favored, as expected, in the suitable scaling regime.

The effective marginals of $\gamma_{N,\mathrm{stat}}$ are obtained as follows. Let $\mu_k$ be the limit
measure corresponding to the simple condensate $\varphi_k^{\otimes N}$. It follows that the measure
$\mu_\mathrm{stat}$ associated to $\gamma_{N,\mathrm{stat}}$ is
\begin{equation*}
  \mu_\mathrm{stat}=\sum_{k=1}^{M}\pi_k\mu_k,
\end{equation*}
and therefore the effective $p$-RDM of $\gamma_\infty^{(p)}$ are of the form
\begin{equation} \label{eq:infty_p_stat} \gamma_{\infty,\mathrm{stat}}^{(p)}=\sum_{k=1}^{M}\pi_k
  \bigl\lvert \,\underbrace{\varphi_k\otimes \dotsm\otimes \varphi_k}_{p} \,\bigr\rangle\bigl\langle\,
  \underbrace{\varphi_k\otimes \dotsm\otimes \varphi_k}_{p} \,\bigr\rvert.
\end{equation}
Notice that due to \eqref{eq:f_pi} we have $\mathrm{Tr}\,\gamma_{\infty,\mathrm{stat}}^{(p)}=1$. This
proves the following Proposition.

\begin{proposition}[\textbf{Rank of effective density matrices of statistical mixtures}]\mbox{}\\
  \label{prop:1}
  Given $\{\varphi_k,f_k(N)\}_{k=1}^{M}$ as in Assumption \ref{assum:one_body}, let us form the state
  $\gamma_{N,\mathrm{stat}}$ as in \eqref{eq:16}. Let $R_{\mathrm{stat}}(p)$ be the rank of the effective
  p-RDM $\gamma^{(p)}_{\infty,\mathrm{stat}}$ associated with $\gamma_{N,\mathrm{stat}}$. Then
  $R_{\mathrm{stat}}(p)$ is a constant function of $p\geq 1$ which coincides with the cardinality $\Pi$ of
  $P$ from \eqref{eq:P} and \eqref{eq:18}
\end{proposition}

In order to compute the effective p-RDM of $\gamma_{N, \mathrm{frag}}$, let $\mu_{\mathrm{frag}}$ be the
probability measure corresponding, in the limit $N\to \infty$, to $\gamma_{N, \mathrm{frag}}$. Let $P$ be as in
\eqref{eq:P} the set of indices corresponding to non-zero $\pi_k$'s and let us define, for any one particle
wave function $\varphi\in \mathcal{H}$, $\delta^{S^1}_u$ as the convex combination of delta measures
\begin{equation*}
  \delta^{S^1}_\varphi=\frac{1}{2\pi}\int_0^{2\pi}\delta_{e^{i\theta}\varphi}  \mathrm{d}\theta.
\end{equation*}
Then (see, \emph{e.g.}, \cite{ammari2016cms})
\begin{equation}
  \label{eq:19}
  \mu_{\mathrm{frag}}=\bigotimes_{k\in P} \delta^{S^1}_{\sqrt{\pi_k}\varphi_k}\otimes \delta^{\perp}_0\; ,
\end{equation}
where $\delta^{\perp}_0$ is the delta measure concentrated in zero and acting on the orthogonal complement
$\mathcal{H}_{\varphi_k}^{\perp}$ of the linear span
\begin{equation*}
  \mathcal{H}_{\varphi_k}= \mathrm{span}_{\mathbb{C}}\bigl\{\varphi_k\;|\; k\in P\bigr\}\; .
\end{equation*}
The crucial feature of the measure $\mu_\mathrm{frag}$ is that it is a measure on a \emph{polycircle},
\emph{i.e.}  it involves the average over more than one independent phase. In fact,
in simple BEC ($\Pi=1$), the $\mathrm{U}(1)$-invariance of the theory masks the effect of the single phase:
indeed, the corresponding marginals coincide with those yielded by a measure concentrated on the
single vector $\varphi_k$ surviving in the limit.

As a consequence, the explicit form of the $p$-RDM
\begin{equation} \label{eq:frag_infty_p_matrix} \gamma_{\infty,\mathrm{frag}}^{(p)}=\int_{\mathcal{H}}^{}\, \bigl\lvert \,\underbrace{\varphi\otimes \dotsm\otimes\varphi}_{p} \,\bigr\rangle\bigl\langle\, \underbrace{\varphi\otimes \dotsm\otimes \varphi}_{p}\,\bigr\rvert\,\mathrm{d}\mu_{\mathrm{frag}}(\varphi)
\end{equation}
associated with $\gamma_{N,\mathrm{frag}}$ is more complicated than that of $\gamma^{(p)}_{\infty,\mathrm{stat}}$. As an
example, the first marginal can be written as follows:
\begin{equation}
  \label{eq:2}
  \gamma_{\infty,\mathrm{frag}}^{(1)}=\frac{1}{(2\pi)^{\Pi}}\int_0^{2\pi}\bigl\lvert\psi_{\theta}\bigr\rangle \bigl\langle \psi_{\theta}\bigr\rvert  \prod_{k\in P}\mathrm{d}\theta_k\;,
\end{equation}
where
\begin{equation}
  \label{eq:3}
  \psi_{\theta}=\sum_{k\in P}^{}e^{i\theta_k}\sqrt{\pi_k}\varphi_k\; .
\end{equation}
The above discussion is summarized by the following Proposition.

\begin{proposition}[\textbf{Effective density matrices of fragmented states}]\mbox{}\\
  \label{prop:matrices}
  Given $\{\varphi_k,f_k(N)\}_{k=1}^{M}$ as in Assumption \ref{assum:one_body}, let us form the state
  $\gamma_{N,\mathrm{frag}}$ as in \eqref{eq:gamma_frag}. Let $\gamma_{\infty,\mathrm{frag}}^{(p)}$,
  $p\geq 1$ be, as in \eqref{eq:frag_infty_p_matrix}, the effective $p$-RDM associated with
  $\gamma_{\infty,\mathrm{frag}}^{(p)}$. Define moreover, for any $\alpha\geq 1$,
  \begin{equation*}
    \begin{split}
      \mathcal{F} _{p,\alpha}:= \Bigl\{g\in \mathbb{N} ^\alpha:\ \sum _{j=1} ^\alpha g _j=p\Bigr\},\ \ \ c _{p,g}:=p!\prod _{\substack{j\in \left\{1,\ldots,M\right\}\\g _j\neq 0}}\frac{\pi _j ^{g _j}}{g _j!}\; .
    \end{split}
  \end{equation*}
  Then,
  \begin{equation*}
    \begin{split}
      \gamma_{\infty,\mathrm{frag}}^{(p)}=\sum _{g \in \mathcal{F} _{p,M}}c _{p,g}&\bigl\lvert\, (\underbrace{\varphi_1\otimes \dotsm\otimes \varphi_1}_{g _1})\vee \dotsm\vee (\underbrace{\varphi_{M}\otimes \dotsm\otimes \varphi_{M}}_{g _{M}})\,\bigr\rangle\\
      &\qquad\bigl\langle\,(\underbrace{\varphi_1\otimes \dotsm\otimes \varphi_1}_{g _1})\vee \dotsm\vee (\underbrace{\varphi_{M}\otimes \dotsm\otimes \varphi_{M}}_{g _{M}})\,\bigr\rvert\; .
    \end{split}
  \end{equation*}
\end{proposition}

We prove Proposition \ref{prop:matrices} in Section \ref{sec:proof-crefthm:matrices}. An immediate
consequance of Proposition \ref{prop:matrices} and \eqref{eq:infty_p_stat} we have the following results.

\begin{corollary}[\textbf{Coincidence of effective 1-RDM's}]\mbox{}\\
  \label{cor:2}
  Under the same assumptions as in Propositions \ref{prop:1} and \ref{prop:matrices}, assume also that the
  two corresponding families $\{\varphi_k,f_k(N)\}_{k=1}^{M}$ coincide. Then
  \begin{equation}
    \gamma_{\infty,\mathrm{frag}}^{(1)}=\gamma_{\infty, \mathrm{stat}}^{(1)}.
  \end{equation}
\end{corollary}

\begin{corollary}[\textbf{Rank of effective density matrices of fragmented states}]\mbox{}\\
  \label{cor:1}
  Let $R_{\mathrm{frag}}(p)$ be the rank of $\gamma_{\infty,\mathrm{frag}}^{(p)}$, and let $\Pi$ be
  defined by \eqref{eq:18}. Then
  \begin{equation*}
    R_{\mathrm{frag}}(p)=\binom{p+\Pi-1}{p}\; .
  \end{equation*}
\end{corollary}

\subsection{Validity of Definition \ref{def:fragm_nostra}}

We exhibited two classes of very different states $\gamma_{N,\mathrm{stat}}$ and $\gamma_{N,\mathrm{frag}}$. They
both satisfy the usual Definition \ref{def:penr_ons} of finite fragmentation. Moreover, if the families
$\{\varphi_k,f_k(N)\}_{k=1}^M$ of Proposition \ref{prop:1} and \ref{prop:matrices} coincide, then by Corollary
\ref{cor:2} the effective 1-RDM of $\gamma_{N,\mathrm{stat}}$ and $\gamma_{N,\mathrm{frag}}$ coincide. The higher
$p$-RDM's however behave quite differently if $\Pi>1$ (if $\Pi=1$ both $\gamma_{N, \mathrm{stat}}$ and $\gamma_{N,
  \mathrm{frag}}$ are pure states describing simple condensates).

By Proposition \ref{prop:1}, the rank of $\gamma_{\infty,\mathrm{stat}}^{(p)}$ is a constant function of $p$. For
this reason $\gamma_{\infty,\mathrm{stat}}$ is \emph{not} fragmented according to our Definition
\ref{def:fragm_nostra}, in agreement with the fact that $\gamma_{N,\mathrm{stat}}$ will be observed with
probability one as a simple condensate. On the other hand, Proposition \ref{prop:matrices} shows that the
rank of $\gamma_{\infty,\mathrm{frag}}^{(p)}$ increases in $p$. This is precisely in accordance with our Definition
\ref{def:fragm_nostra}.

The characterization we propose has therefore the feature of excluding at least statistical mixtures of
simple condensates. This comes at the expense of Definition \ref{def:fragm_nostra} being only slightly
more difficult to verify than Definition \ref{def:penr_ons} in concrete cases. It is indeed worth
remarking that all physically relevant states with finite fragmented condensation that we know of satisfy
our Definition \ref{def:fragm_nostra}. To mention a concrete example, this includes the spin-one
fragmented state corresponding to the LPB wave function~\cite{law1998prl}. The LPB wavefunction is the
ground state for a system of $N$ weakly interacting spin-one bosons, with total spin zero. It can be
written mathematically as
\begin{equation*}
  {\bigl(a^{*}_{0,1}a^{*}_{0,-1}+a^{*}_{0,-1}a^{*}_{0,1}-a^{*}_{0,0}a^{*}_{0,0} \bigr)}^{\frac{N}{2}}\lvert \mathrm{vac}\rangle \; ,
\end{equation*}
where $a_{k,s}^{\sharp}$ are the creation and annihilation operators corresponding to momentum $k\in \mathbb{R}^3$ and
spin $s\in \{-1,0,1\}$ in the third spatial direction, and $\lvert \mathrm{vac}\rangle$ is the Fock vacuum
vector. The LPB state is not of the type $\gamma_{N, \mathrm{frag}}$ previously
considered. Nonetheless its effective 2-RDM is explicitly computable, and the resulting rank is larger
than the one of the 1-RDM.

To sum up, we believe that our Definition \ref{def:fragm_nostra} characterizes finite fragmented
condensation better than the usual Penrose-Onsager-like definition, at the same time being only slightly
more difficult to verify.

\section{Persistence of finite fragmented condensation}
\label{sec:pers-finite-fragm}

In this Section we study the behavior of (temperature-zero) finite fragmented condensation under a time
evolution ruled by a Hamiltonian with two-body interactions in the mean-field regime.

We will consider a system of $N$ bosons in $\mathbb{R}^d$ with (pseudo-) spin $s$, thus specializing the
notations of the previous Section to
\begin{equation}
  \mathcal{H}=L^2(\mathbb{R}^d,\mathbb{C}^{2s+1}),\qquad\cH_N:=\bigvee_{j=1}^N L^2(\mathbb{R}^d,\mathbb{C}^{2s+1}),\qquad s\in \frac{1}{2}\mathbb{N}\smallsetminus \{0\}.
\end{equation}
Let us consider the many-body Hamiltonian
\begin{equation}
  \label{eq:mf_Hamiltonian}
  H_{\omega,N}:=\sum_{j=1}^N \,h_{\omega,\,j}+\frac{1}{N}\sum_{j<k} V(x_j-x_k)\; ,
\end{equation}
acting on $\mathcal{H}_N$. The one-body Hamiltonian $h_\omega$ and the pair potential $V$ satisfy the
assumptions below. The notation $h_{\omega,j}$ indicates the operator which acts as $h_\omega$ on the
$j$-th copy $L^2(\mathbb{R}^d,\mathbb{C}^{2s+1})$ inside $\mathcal{H}_N$, and as the identity on all the
other factors. The choice of the coupling constant $N^{-1}$ in the pair interaction effectively puts us in
the so-called \emph{mean-field} regime, by making, at least formally, the two contributions to
$H_{\omega,N}$ of the same order.

\begin{assumption}[\textbf{The one-particle Hamiltonian $h_\omega$}]\mbox{}\\ \label{assum:h}
  The operator $h_\omega$ is positive and acts on $\mathcal{H}$, with the following properties:
  \begin{itemize}
    \setlength{\itemsep}{2mm}
    \item \textbf{(Domain):} $\mathcal{D}(h_\omega)$ does not depend on $\omega$;
  \item \textbf{(No action on spin):} $h_{\omega}=\underline{h}_{\omega}\otimes
    \mathrm{id}_{\mathbb{C}^{2s+1}}$, for a positive operator $\underline{h}_{\omega}$ on
    $L^2(\mathbb{R}^d,\mathbb{C})$;
  \item \textbf{(Zero ground state energy):} $\inf\sigma(\underline{h}_\omega)=0$, where
    $\sigma(\underline{h}_\omega)$ is the spectrum of $\underline{h}_\omega$;
  \item \textbf{(Ground state):} the set of $\underline{\varphi}\in L^2(\mathbb{R}^d,\mathbb{C})$, such that
    \begin{equation*}
      \underline{h}_\omega\underline{\varphi}=0
    \end{equation*}
    is a one-dimensional subspace, that does not depend on $\omega$.
  \item \textbf{(Gap condition):} $\inf
    \Bigl(\sigma(\underline{h}_{\omega})\smallsetminus\{0\}\Bigr)=\omega\in(0,+\infty)$.
  \end{itemize}
\end{assumption}

\begin{assumption}[\textbf{The pair potential $V$}]\mbox{}\\ \label{assum:V}
  The pair potential $V\in L^2_{\mathrm{loc}}(\mathbb{R}^d,\mathbb{C}^{2s+1})$ is an even function which is controlled by the one-body operator
  $h_{{\omega}}$ for every $\omega>0$ in the following sense: for every $\varepsilon>0$, there exists a
  constant $C_\varepsilon$ such that
  \begin{equation} \label{eq:bound_V1} |V(x-y)|\le \varepsilon h_{{\omega},x}\otimes \mathcal{I}_y +C_\varepsilon
  \end{equation}
  as quadratic forms on the product of form domains $\mathcal{D}[h_{{\omega}}]\otimes\mathcal{D}[h_{{\omega}}]$,
  \emph{and}
  \begin{equation}\label{eq:bound_V2}
    \|V\psi\|\le \varepsilon\| h_{{\omega}}\psi\|+C_\varepsilon\|\psi\|
  \end{equation}
  for every $\psi\in\mathcal{D}(h_{{\omega}})$.
\end{assumption}

Even though our assumptions are stated in a rather technical fashion, which makes them ready to be applied
in the proofs, they are fulfilled by a large class of Hamiltonians of physical interest. This is the case,
for example, of a one-particle Hamiltonian with kinetic energy operator and harmonic trap in
$d$-dimensions, \emph{i.e.},
\begin{equation} \label{eq:harmonic} \underline{h}_\omega=\omega(-\Delta+x^2-d)
\end{equation}
(we subtracted the term $d$ in order to have zero ground state energy, and $\omega$ plays the role of a
dilation factor). In this latter case, and for $d=3$, we can also consider the physically relevant case of
a pair interaction with a local Coulomb singularity, that is, $V(x)= c|x|^{-1}$. This is possible because
\eqref{eq:bound_V1}, \eqref{eq:bound_V2} are a consequence of Hardy's inequality
\begin{equation*}
  \frac{1}{4|x|^2}\le-\Delta_x \quad\text{on }L^2(\mathbb{R}^3).
\end{equation*}
As an immediate consequence, and again in the case of \eqref{eq:harmonic} in $d=3$, we can also consider
$V$ with local singularities of type $|x|^{-\alpha}$ with $\alpha\in[0,1)$. For generic $h_\omega$ and for every space
dimension $d$, we note that any $V\in L^\infty(\mathbb{R}^d)$ satisfies Assumption \ref{assum:V} since
\eqref{eq:bound_V1}, \eqref{eq:bound_V2} are trivially true even for $\varepsilon=0$.

Assumption \ref{assum:V} imply that, by Kato-Rellich Theorem, the many-body Hamiltonian $H_{\omega,N}$ defined
in \eqref{eq:mf_Hamiltonian} is self-adjoint on $D\Bigl(\sum_{j=1}^Nh_{\omega,j}\Bigr)$. Moreover, $h_{\omega}$ has a
$2s+1$-fold degenerate ground state, of energy equal to zero, spanned by the orthonormal functions
\begin{equation} \label{eq:gss}
  \varphi_1=\bigl(\underline{\varphi},0,\dotsc,0\bigr)\,,\dotsm,\,\varphi_{2s+1}=\bigl(0,\dotsc,0,\underline{\varphi}\bigr).
\end{equation}
The degeneracy is induced by the degrees of freedom due to the particles' (pseudo-) spin. Let us denote
\begin{equation}
  \label{eq:1}
  \mathcal{F}:= \operatorname{span}_{\mathbb{C}}\Bigl\{\varphi_1,\dotsc,\varphi_{2s+1}\Bigr\}\subset \mathcal{H}\; .
\end{equation}
For any $k=1,\dots,2s+1$, $\psi\in \mathcal{F}^{\perp}$, $\lVert \psi \rVert_{\mathcal{H}}^{}=1$ we have
the inequality
\begin{equation*}
  \langle \psi  , h_{\omega}\psi \rangle_{\mathcal{H}}-\langle \varphi_k  , h_{\omega}\varphi_{k} \rangle_{\mathcal{H}}\geq \omega\;.
\end{equation*}
This shows that the one-particle Hamiltonian $h_{\omega}$ also has an energy gap $\omega$.

Given an initial many-body configuration $\Psi_0\in\cH_N$, the time-evolution is ruled by the Schrödinger
equation
\begin{equation} \label{eq:N_body_schroedinger} \left\{\begin{aligned}
      &i\partial_t\Psi(t)=H_{\omega,N}\Psi(t)\\
      &\Psi(0)=\Psi_{0}
    \end{aligned}\right . \; ,
\end{equation}
whose solution is $\Psi(t)=e^{-itH_{\omega,N}}\Psi_0$. Since we aim at studying the time evolution of
fragmented condensates, we will consider as initial datum a pure ground state of $\sum_j h_{\omega,j}$ in
a fragmented BEC phase according to our Definition \ref{def:fragm_nostra}, \emph{i.e.},
\begin{equation} \label{eq:initial_datum}
  \gamma_{N,0}\equiv\lvert\Psi_0\rangle\langle\Psi_{0}\rvert\qquad\text{with
  }\qquad\Psi_0=\varphi_1^{\otimes f_1(N)}\vee\dotsm\vee \varphi_{2s+1}^{\otimes f_{2s+1}(N)}.
\end{equation}
Here, as in Section \ref{sec:char-fragm-cond}, we assume that the $f_k(N)$'s are (a sequence of) integers
such that
\begin{equation} \label{eq:f_pi'} \sum_{k=1}^{2s+1}f_k(N)=N,\;\;\forall
  N\qquad\text{and}\qquad\exists\,\pi_k:=\lim_{N\to\infty}\frac{f_k(N)}{N}\in[0,1].
\end{equation}
In other words, the family $\{\varphi_k,f_k(N)\}_{k=1}^{2s+1}$ with $\varphi_k$ defined in \eqref{eq:gss}
satisfies Assumption \ref{assum:one_body}.

The state $\gamma_{N,0}$ is indeed of the type \eqref{eq:gamma_frag} (thence it is fragmented according to
Definition \ref{def:fragm_nostra}) and, as we discussed in Section \ref{sec:char-fragm-cond}, the $N\to
\infty$ counterpart of $\gamma_{N,0}$ is a probability measure $\mu_0$ on the one-particle space
$\mathcal{H}$, in the sense that the $p$-th effective marginal associated with $\gamma_{N,0}$ is
\begin{equation}
  \label{eq:6}
  \gamma_{\infty,0}^{(p)}=\int_{\mathcal{H}}^{}\,\lvert \psi^{\otimes p}\rangle\langle \psi^{\otimes p}\rvert\;  \mathrm{d}\mu_0(\psi)\; .
\end{equation}
Explicitly, the measure $\mu_0$ is the $U(1)$-invariant product of convex combinations of delta measures
\begin{equation}
  \label{eq:7}
  \mu_0=\bigotimes_{k\in P}^{}\delta^{S_1}_{\sqrt{\pi_{k}}\varphi_{k}} \otimes \delta_0^{\perp}
\end{equation}
where $P$ is, as in \eqref{eq:P},
\begin{equation*}
  P=\Bigl\{k=1,\dotsc,2s+1\;|\; \pi_k\neq 0\Bigr\} \qquad\text{and}\qquad\Pi=\mathrm{Card}(P)\; .
\end{equation*}

The time-evolution induced by \eqref{eq:N_body_schroedinger} on the measure $\mu_0$ is given in the next
result, which we import directly from~\cite{ammari2015asns}.

\begin{theorem}[\textbf{Mean-field evolution of fragmented states}]\mbox{}\\ \label{thm:mf_fragm}
  Suppose Assumptions \ref{assum:h} and \ref{assum:V} hold. Let $\{\varphi_k,f_k(N)\}_{k=1}^N$ satisfy
  Assumption \ref{assum:one_body}, with $\varphi_k$ defined in \eqref{eq:gss}. Let $\Psi(t)$ be the
  solution to \eqref{eq:N_body_schroedinger} with initial datum
  \begin{equation*}
    \Psi_0=\varphi_1^{\otimes f_1(N)}\vee\dotsm\vee \varphi_{2s+1}^{\otimes f_{2s+1}(N)}
  \end{equation*}	
  and consider the evolved state $\gamma_{N,t}=|\Psi(t)\rangle\langle\Psi(t)|$.  Then for any fixed
  $p\in\mathbb{N}$ there exists the limit of the $p$-RDM
  \begin{equation*}
    \gamma^{(p)}_{\infty,t}=\lim_{N\to\infty}\gamma^{(p)}_{N,t}
  \end{equation*}
  in trace-norm, and
  \begin{equation}
    \label{eq:9}
    \gamma_{\infty,t}^{(p)}=\int_{\mathcal{H}}^{}\,\lvert \psi(t)^{\otimes p}\rangle\langle \psi(t)^{\otimes p}\rvert\;  \mathrm{d}\mu_0(\psi).
  \end{equation}
  Here $\mu_0$ is the initial measure \eqref{eq:7}, while $\psi(t)$ is the unique solution of the
  effective Hartree Cauchy problem
  \begin{equation}
    \label{eq:10}
    \begin{cases}
      \;i\partial_t\psi(t)=h_{\omega}\psi(t) + \bigl(V*\lvert \psi(t)  \rvert_{}^{2}\bigr)\psi(t)\\
      \;\psi(0)=\psi.
    \end{cases}
  \end{equation}
  where the initial datum $\psi$ is the integration variable in \eqref{eq:9}.
\end{theorem}

The above result shows how the effective probability distribution of single-particle states is pushed
forward by the Hartree effective evolution. Let us remark that the error made in approximating the evolved
interacting $N$-particle reduced density matrices with $\gamma_{\infty,t}^{(p)}$ given by \eqref{eq:9} is
of order $N^{-1}$ for any fixed time (with a deterioration as $t$ increases). This is confirmed by
theoretical and numerical analysis \cite{ammari2016cms}.

From this mathematical description it is quite easy to understand why, generically, inter-particle
interaction \emph{breaks down finite fragmentation}. Let us consider for simplicity the case of
\begin{equation*}
  \Psi_0=\varphi_1^{\otimes N/2}\vee \varphi_2^{\otimes N/2}.
\end{equation*}
As already discussed, it is convenient to rewrite \eqref{eq:6} for $p=1$ as
\begin{equation*}
  \begin{split}
    \gamma_{\infty,0}^{(1)}=\;&\frac{1}{(2\pi)^2}\int_0^{2\pi}\int_0^{2\pi}\Big| \frac{e^{i\theta_1}}{\sqrt{2}}\varphi_1+\frac{e^{i\theta_2}}{\sqrt{2}}\varphi_2\Big\rangle \Big\langle \frac{e^{i\theta_1}}{\sqrt{2}}\varphi_1+\frac{e^{i\theta_2}}{\sqrt{2}}\varphi_2\Big| d\theta_1 d\theta_2\\
    =\;&\frac{1}{2}|{\varphi_1}\rangle\langle\varphi_1|+\frac{1}{2}|{\varphi_2}\rangle\langle\varphi_2|.
  \end{split}
\end{equation*}
Notice in particular that the crossed terms vanish exactly, and of course there is no contribution coming
from the sector of the Hilbert space orthogonal to $\varphi_1$ and $\varphi_2$. At later times, however,
by Theorem \ref{thm:mf_fragm} and \eqref{eq:7} we have
\begin{equation} \label{eq:gamma1t}
  \begin{split}
    \gamma_{\infty,t}^{(1)}=\;\frac{1}{(2\pi)^2}\int_0^{2\pi}\int_0^{2\pi}| \psi_{\theta_1,\theta_2}(t)\rangle \langle \psi_{\theta_1,\theta_2}(t)| d\theta_1 d\theta_2
  \end{split}
\end{equation}
with
\begin{equation*}
  \begin{cases}
    \;i\partial_t\psi_{\theta_1,\theta_2}(t)=h_{\omega}\psi_{\theta_1,\theta_2}(t) + \bigl(V*\lvert \psi_{\theta_1,\theta_2}(t)  \rvert_{}^{2}\bigr)\psi_{\theta_1,\theta_2}(t)\\
    \;\psi_{\theta_1,\theta_2}(0)=\frac{e^{i\theta_1}}{\sqrt{2}}\varphi_1+\frac{e^{i\theta_2}}{\sqrt{2}}\varphi_2.
  \end{cases}
\end{equation*}
Due to the nonlinearity, for $t>0$ the $\theta_1,\theta_2$-average of the projections onto the above Hartree
solutions generically has infinite rank. This breaks down fragmentation according
to both definitions presented in Section \ref{sec:char-fragm-cond}. The situation is different in the sole
case $\Pi=1$, \emph{i.e.}\ when the initial state is actually a simple condensate. In this case the measure
$\mu_0$ is a $U(1)$-invariant convex combination of delta measures, and such structure is preserved by the
$U(1)$-invariant nonlinear Hartree evolution; simple condensation is therefore preserved, as already known
from the rich mathematical literature on the evolution of BEC's in the mean-field
regime~\cite{ginibre1979cmpI,ginibre1979cmp2,rodnianski2009cmp,ErdSchYau-10,knowles2010cmp,Benedikter-Porta-Schlein-2015}.

We will show in the following that we can recover persistence of finite fragmented BEC if, within our
Assumption \ref{assum:h}, the gap $\omega$ between the degenerate ground state and the first excited states of
the one-particle Hamiltonian becomes very large. The intuitive explanation is the following: the energetic
cost to transition from the (free) ground state to excited states is so high that inter-particle
interactions are not strong enough to cause such a jump. In this way particles are effectively constrained
to the $2s+1$-dimensional Hilbert space of degenerate ground states of $h_{\omega}$, and this preserves the
finite fragmentation caused by spin degeneracy. Nonetheless, there is still an effective nonlinear
evolution occurring in the aforementioned $2s+1$-dimensional Hilbert space.

We will rigorously justify the above intuition, and provide an explicit effective one-particle evolution
valid in the span of the degenerate ground states in the limit $\omega\to \infty$. For the sake of simplicity we will
focus on the evolution of the first marginal $\gamma_{\infty,t}^{(1)}(\omega)$ (where the dependence on $\omega$ is made
explicit to clarify that we are studying the infinite gap limit). At the end we will briefly discuss the
behavior of higher marginals.

\begin{theorem}[\textbf{Persistence of fragmented BEC in the large gap limit}]\mbox{}\\
  \label{thm:convergence}
  Let $\{\kappa_{j,t,\{\theta_{m},m\in P\}}(\infty)\}_{j=1}^{2s+1}$ solve the system of ODEs below, for $t\ge0$,
  \begin{equation} \label{eq:limit_eq} \left\{
      \begin{aligned}
        &\;i \partial_t \kappa_{j,t,\{\theta_{m},m\in P\}}(\infty)=\Bigl\langle \;\varphi_{j}\;  ,\; V*\Bigl\lvert \sum_{\ell=1}^{2s+1}\kappa_{\ell,t,\{\theta_{m},m\in P\}}(\infty)\varphi_{\ell}  \Bigr\rvert_{}^2\Bigl(\sum_{\ell'=1}^{2s+1}\kappa_{\ell',t,\{\theta_{m},m\in P\}}(\infty)\varphi_{\ell'} \Bigr)\; \Bigr\rangle_{} \\
        &\;\kappa_{j,0,\{\theta_{m},m\in P\}}(\infty)=\left\{
          \begin{aligned}
            &e^{i\theta_{j}}\sqrt{\pi_{j}}& \text{if } j\in P\\
            &0& \text{if } j\notin P,
          \end{aligned}
        \right .
      \end{aligned}
    \right.
  \end{equation}
  and define the infinite gap phase averages
  \begin{equation}
    \label{eq:large_gap_coefficient}
    \begin{split}
      K_{j\ell,t}(\infty)&=\frac{1}{(2\pi)^\Pi}\int_0^{2\pi}  \overline{\kappa_{\ell,t,\{\theta_{m},m\in P\}}(\infty)}\kappa_{j,t,\{\theta_{m},m\in P\}}(\infty)  \, \prod_{m\in P}^{} \mathrm{d}\theta_m.
    \end{split}
  \end{equation}
  Define also the effective 1-RDM in the large gap limit
  \begin{equation*}
    \gamma^{(1)}_{\infty,t}(\infty):=\sum_{j=1}^{2s+1}K_{jj,t}(\infty)\lvert\varphi_{j}\rangle\langle\varphi_{j}\rvert+\sum_{j<\ell=1}^{2s+1}\Big(K_{j\ell,t}(\infty)\lvert\varphi_{j}\rangle\langle\varphi_{\ell}\rvert+\overline{K_{j\ell,t}(\infty)}\lvert\varphi_{\ell}\rangle\langle\varphi_{j}\rvert\Big).
  \end{equation*}
  Suppose that Assumptions \ref{assum:h} and \ref{assum:V} hold, and let $\gamma^{(1)}_{N,t}(\omega)$ be
  the 1-RDM associated to $\Psi(t)$ solution to \eqref{eq:N_body_schroedinger} for $t\ge0$. Let the
  initial datum $\Psi_0$ be given by \eqref{eq:initial_datum} with $\{\varphi_k,f_k(N)\}_{k=1}^{2s+1}$
  satisfying Assumption \ref{assum:one_body} and $\varphi_k$ defined in \eqref{eq:gss}. Then for any fixed
  $t\ge0$ we have\begin{equation*}
    \gamma^{(1)}_{\infty,t}(\infty)=\lim_{\omega\to\infty}\lim_{N\to\infty}\gamma_{N,t}^{(1)}(\omega)
    =\lim_{N\to\infty}\lim_{\omega\to\infty}\gamma_{N,t}^{(1)}(\omega)
  \end{equation*}
  in trace-norm.  As a consequence
  \begin{equation*}
    \operatorname{Rank} \gamma^{(1)}_{\infty,t}(\infty)=  2s+1 \text{ for almost all $t\in \mathbb{R}$}\,.
  \end{equation*}
\end{theorem}

We will prove Theorem \ref{thm:convergence} in Section \ref{sec:proof-crefthm:convergence1} and Section
\ref{sec:proof-crefthm:convergence2}. Let us remark that it is possible to extend Theorem
\ref{thm:convergence} without difficulty to any $p$-RDM.  In fact, one can deduce that the $p$-RDM has the
following form
\begin{equation*}
\gamma_{\infty,t}^{(p)}(\infty)=\frac{1}{(2\pi)^{\Pi}}\int_0^{2\pi} \biggl(\Bigl\lvert\sum_{j=1}^{2s+1}\kappa_{j,t,\{\theta_m,m\in P\}}(\infty)\varphi_j \Bigr\rangle\Bigl\langle \sum_{j=1}^{2s+1}\kappa_{j,t,\{\theta_m,m\in P\}}(\infty)\varphi_j\Bigr\rvert\biggr)^{\otimes p} \prod_{m\in P} \mathrm{d}\theta_m  \; ,
\end{equation*}
and therefore
\begin{equation*}
  \operatorname{Rank} \gamma^{(p)}_{\infty,t}(\infty)=  \binom{p+2s}{p}\text{ for almost all $t\in \mathbb{R}$}\; .
\end{equation*}
This shows that fragmented BEC in the sense of our Definition \ref{def:fragm_nostra} is preserved by the
mean-field evolution in the $\omega\to\infty$ limit. The system remains finitely fragmented on the space of
one-particle ground states: at almost every time there is a non-zero macroscopic fraction of particles
occupying \emph{all} the available degenerate one-particle ground states, provided that at $t=0$ at least
two of them had macroscopic occupation. No macroscopic occupation of the higher energy space of the
one-particle Hamiltonian occurs.

We remark that, even if the two limits $\omega\to\infty$ and $N\to\infty$ yield the same result no matter the order in which they are taken, our strategy does not allow to take the joint limit $N,\omega\to\infty$. This is due to a lack of uniformity in the bounds we use in the proof.

\subsection{Discussion of Theorem~\ref{thm:convergence}, and outline of its proof}

Let us briefly discuss why, at least formally, \eqref{eq:limit_eq} and \eqref{eq:large_gap_coefficient}
should accurately describe the large-$\omega$ behavior of our system. Using \eqref{eq:7} and \eqref{eq:9}
we rewrite \eqref{eq:10} as
\begin{equation}
  \label{eq:11}
  \gamma_{\infty,t}^{(1)}(\omega)=\frac{1}{(2\pi)^\Pi}\int_0^{2\pi} \bigl\lvert \psi^{(\omega)}_{\{\theta_{m},m\in P\}}(t)\bigr\rangle\bigl\langle \psi^{(\omega)}_{\{\theta_{m},m\in P\}}(t)\bigr\rvert \, \prod_{m\in P}^{} \mathrm{d}\theta_m\;,
\end{equation}
where $\psi^{(\omega)}_{\{\theta_{m},m\in P\}}(t)$ is the solution of \eqref{eq:10} with initial condition
\begin{equation}
  \label{eq:12}
  \psi^{(\omega)}_{\{\theta_{m},m\in P\}}(0)=\sum_{k\in P}^{}e^{i\theta_k}\sqrt{\pi_k}\varphi_k\; .
\end{equation}
Let us decompose $\psi^{(\omega)}_{\{\theta_{m},m\in P\}}(t)$ according to the orthogonal Hilbert space
decomposition
\begin{equation*}
  \mathcal{H}=\mathcal{F} \oplus \mathcal{F}^{\perp}\; ,
\end{equation*}
with $\mathcal{F}$ introduced in \eqref{eq:1}. Let us stress that $\mathcal{F}$ is spanned by \emph{all}
the $2s+1$ degenerate ground states, even those for which $\pi_k=0$. We obtain
\begin{equation}
  \label{eq:13}
  \psi^{(\omega)}_{\{\theta_{m},m\in P\}}(t)=\sum_{k=1}^{2s+1}\kappa_{k,t\{\theta_{m},m\in P\}}(\omega)\varphi_{k}+\psi^{\perp}_{t}(\omega)\; ;
\end{equation}
where $\kappa_{k,t\{\theta_{m},m\in P\}}(\omega)\in \mathbb{C}$ for any $k=1,\dotsc,2s+1$, and
$\psi^{\perp}_{t}(\omega)\in \mathcal{F}^{\perp}$.  A direct computation shows that the coefficients
$\kappa_{k,t\{\theta_{m},m\in P\}}(\omega)$ satisfy the system of ODE's (keeping in mind that $\psi^{(\omega)}_{\{\theta_{m},m\in P\}}(t)$ is a solution of the Hartree equation \eqref{eq:10})
\begin{equation} \label{eq:ODE_omega} \left\{
    \begin{aligned}
      &i \partial_t \kappa_{j,t,\{\theta_{m},m\in P\}}(\omega)=\biggl\langle \;\varphi_{j}\;  ,\; V*\bigl\lvert\psi^{(\omega)}_{\{\theta_{m},m\in P\}}(t) \bigr\rvert_{}^2\bigl(\psi^{(\omega)}_{\{\theta_{m},m\in P\}}(t)\bigr)\; \biggr\rangle_{} \\
      &\kappa_{j,0,\{\theta_{m},m\in P\}}(\omega)=\left\{
	\begin{aligned}
          &e^{i\theta_{j}}\sqrt{\pi_{j}}& \text{if } j\in P\\
          &0& \text{if } j\notin P
	\end{aligned}
      \right .
    \end{aligned}
  \right .\quad .
\end{equation}
	
In the large-$\omega$ limit, it should not be possible for the Hartree flow to transfer any fraction of
the norm of $\psi^{(\omega)}_{\{\theta_{m},m\in P\}}(t)$ from $\mathcal{F}$ to
$\mathcal{F}^\perp$. Intuitively this should mean
\begin{equation*}
  \psi_t^\perp(\omega)\simeq 0,\quad\text{as }\omega\to\infty.
\end{equation*}
This will be rigorously proven in Proposition \ref{prop:perp_small}. If we assume
$\psi_t^\perp(\omega)=0$, then \eqref{eq:ODE_omega} formally reduces to \eqref{eq:limit_eq} with
$\omega=\infty$. The interpretation is that the $\mathcal{F}$-components of the Hartree-evolved
wave-functions eventually decouple from the evolution occurring in $\mathcal{F}^\perp$. The limit ODE's
\eqref{eq:limit_eq} show that the evolution of the $\kappa's$ is influenced only by the initial datum
(which also belongs to $\mathcal{F}$).

Let us define for any $j<l=1,\dotsc,2s$ the `phase-averaged' transition amplitudes
\begin{equation}
  \label{eq:14}
  \begin{split}
    K_{j\ell,t}(\omega)&=\frac{1}{(2\pi)^\Pi}\int_0^{2\pi}  \overline{\kappa_{\ell,t,\{\theta_{m},m\in P\}}(\omega)}\kappa_{j,t,\{\theta_{m},m\in P\}}(\omega)  \, \prod_{m\in P}^{} \mathrm{d}\theta_m
  \end{split}.
\end{equation}
For $t=0$, the $\{K_{j\ell,0}(\omega)\}_{j\ell}$'s are precisely the coefficients of the 1-RDM of the
system in the basis $\{\varphi_k\}$. In particular, only the $K_{jj,0}(\omega)$ differ from zero, due to
the phase integration. At later times, and for finite $\omega$, the $\{K_{j\ell,t}(\omega)\}_{j\ell}$'s
still define a positive matrix in the space spanned by the $\varphi_k$'s. This matrix however will in
general not be diagonal, and more importantly, its trace will in general be smaller than one. This is due
to the Hartree flow transfering mass to $\psi^\perp_t(\omega)\in\mathcal{F}^\perp$. In the large $\omega$
limit however, by neglecting every transition to excited states, we recover that the
$\{K_{j\ell,t}(\infty)\}_{j\ell}$'s define a matrix with unit trace which is the effective 1-RDM of the
system.

\section{Proof of Theorem \ref{thm:convergence}, case $N\to\infty$ followed by $\omega\to\infty$}
\label{sec:proof-crefthm:convergence1}

We prove \eqref{thm:convergence} in two steps. In this Section we prove the case of $N\to \infty$ first
and then $\omega\to \infty$. In Section \ref{sec:proof-crefthm:convergence1} we prove the case of limits
in the reverse order, and show that we obtain the same result.

The limit $N\to\infty$ was proven in~\cite{ammari2015asns,ammari2016cms}, and we already presented the
result in Theorem \ref{thm:mf_fragm}. Hence, we consider directly the effective problem and prove that, at
every fixed time $t\ge0$, the component of $\psi^{(\omega)}_{\{\theta_{m},m\in P\}}(t)$ that is orthogonal
to all the $\varphi_k$'s vanishes in $L^2$ for $\omega\to\infty$. Then, we show that the coefficients
$\{\kappa_{j,t,\{\theta_{m},m\in P\}}(\omega)\}$ converge to the solution $\{
\kappa_{j,t,\{\theta_{m},m\in P\}}(\infty)\}$ of \eqref{eq:limit_eq} (the existence and uniqueness of
solutions to this system of ODE does not present difficulties). The above will imply that
$\psi^{(\omega)}_{\{\theta_{m},m\in P\}}(t)$ has a limit for every $\theta_{j}\in[0,2\pi]$ and for every
$t\ge0$, as $\omega\to\infty$. Then we show how this implies the convergence of the 1-RDM. This concludes
the proof of \eqref{thm:convergence} if the limits are taken in the order $N\to \infty$, and then
$\omega\to \infty$.

The first step is the following Proposition.

\begin{proposition}[\textbf{Restriction to ground states of $h_\omega$ in the large-$\omega$ limit}]\mbox{}\\\label{prop:perp_small}
  In the same assumptions of Theorem \ref{thm:convergence}, consider, as in \eqref{eq:13}, the
  decomposition
  \begin{equation*}
    \psi^{(\omega)}_{\{\theta_{m},m\in P\}}(t)=\sum_{k=1}^{2s+1}\kappa_{k,t\{\theta_{m},m\in P\}}(\omega)\varphi_k+\psi^{\perp}_{t}(\omega)\;
  \end{equation*}
  of the solution to the Hartree equation \eqref{eq:10} with initial datum
  \begin{equation*}
    \psi^{(\omega)}_{\{\theta_{m},m\in P\}}(0)=\sum_{k\in P}^{} e^{i\theta_k}\sqrt{\pi_k}\varphi_k\; .
  \end{equation*}
  Then there exists $C>0$ independent of $t$ and $\omega$ such that
  \begin{equation} \label{eq:small_perp} \|\psi^{\perp}_{t}(\omega)\|_2\le \frac{C}{\omega^{1/2}}\;.
  \end{equation}
\end{proposition}
\begin{proof}We will show that the Hartree energy
  \begin{equation}
    \mathcal{E}[\psi]=\langle\psi, h_\omega\psi\rangle+\langle\psi,V*|\psi|^2\psi\rangle
  \end{equation}
  suitably controls the norm $\|\psi^{\perp}_{t}(\omega)\|_2$, and then use conservation of
  $\mathcal{E}[\psi^{(\omega)}_{\{\theta_{m},m\in P\}}(t)]$ along the Hamiltonian flow of \eqref{eq:10}.
	
  By the assumption \eqref{eq:bound_V1} on $V$ we have
  \begin{equation} \label{eq:control_V,h} \big|\langle\psi, V*|\psi|^2\psi\rangle\big|\le
    \langle\psi\otimes\psi,\big[ {\varepsilon} h_{\omega,x}\otimes \mathcal{I}_y+C_\varepsilon \big]\psi\otimes\psi\rangle\le \varepsilon \langle\psi,
    h_\omega\psi\rangle+C_\varepsilon
  \end{equation}
  for every $\psi\in\mathcal{D}[h_\omega]$ with $\|\psi\|_2=1$. This immediately implies
  \begin{equation*}
    (1-\varepsilon) \langle\psi,h_\omega\psi\rangle\le \mathcal{E}[\psi]+C_\varepsilon\; .
  \end{equation*}
  Now, since $\mathcal{E}[\psi^{(\omega)}_{\{\theta_{m},m\in
    P\}}(t)]=\mathcal{E}[\psi^{(\omega)}_{\{\theta_{m},m\in P\}}(0)]$, we deduce
  \begin{equation} \label{eq:upper} (1-\varepsilon)\langle\psi^{(\omega)}_{\{\theta_{m},m\in
      P\}}(t),h_\omega\psi^{(\omega)}_{\{\theta_{m},m\in P\}}(t)\rangle\le
    \mathcal{E}[\psi^{(\omega)}_{\{\theta_{m},m\in P\}}(0)]+C_\varepsilon\; .
  \end{equation}
  Moreover, by Assumption \ref{assum:h} $h_\omega$ vanishes on the ground states $\varphi_j$'s and has gap
  $\omega$, and therefore
  \begin{equation} \label{eq:lower} \langle\psi^{(\omega)}_{\{\theta_{m},m\in
      P\}}(t),h_\omega\psi^{(\omega)}_{\{\theta_{m},m\in P\}}(t)\rangle\ge
    \omega\|\psi^{\perp}_{t}(\omega)\|_2^2\; .
  \end{equation}
  Comparing \eqref{eq:upper} and \eqref{eq:lower}, we deduce that for any $0<\varepsilon<1$,
  \begin{equation}
    \|\psi^{\perp}_t(\omega)\|_2^2\le\,\frac{\mathcal{E}[\psi^{(\omega)}_{\{\theta_{m},m\in P\}}(0)]+C_\varepsilon}{(1-\varepsilon)\omega}\; .
  \end{equation}
  This concludes the proof, since by \eqref{eq:control_V,h} the quantity
  \begin{equation*}
    \begin{split}
      \big|\mathcal{E}[\psi^{(\omega)}_{\{\theta_{m},m\in P\}}(0)]\big|&=\big|\langle \psi^{(\omega)}_{\{\theta_{m},m\in P\}}(0), V*|\psi^{(\omega)}_{\{\theta_{m},m\in P\}}(0)|^2\psi^{(\omega)}_{\{\theta_{m},m\in P\}}(0)\rangle\big|
    \end{split}
  \end{equation*}
  is bounded by a constant independent of $\omega$.
\end{proof}

With an estimate of smallness of $\psi^\perp_t(\omega)$ in $L^2$ available, our next step is to show that
$\kappa_{j,t,\{\theta_{m},m\in P\}}(\omega)$ converges to $\kappa_{j,t,\{\theta_{m},m\in P\}}(\infty)$.

\begin{proposition}[\textbf{Convergence of orthogonal components}]\mbox{}\\ \label{prop:coefficient_convergence} In the same assumptions of Theorem \ref{thm:convergence}, let $\kappa_{j,t,\{\theta_{m},m\in P\}}(\omega)$, with $j=1,\dots,2s+1$, be the components of $\psi^{(\omega)}_{\{\theta_{m},m\in P\}}(t)$ defined by the orthogonal decomposition \eqref{eq:13}, and let $\kappa_{j,t,\{\theta_{m},m\in P\}}(\infty)$ be the solution to \eqref{eq:limit_eq}. Then for any $t\ge0$ we
  have
  \begin{equation} \label{eq:coefficient_convergence} |\kappa_{j,t,\{\theta_{m},m\in
      P\}}(\omega)-\kappa_{j,t,\{\theta_{m},m\in P\}}(\infty)|^2\le \frac{1}{\omega} e^{C|t|}\; ,
  \end{equation}
  for some $C>0$ independent on $t$ or $\omega$. As a consequence,
  \begin{equation} \label{eq:L_2_convergence} \bigg\|\;\psi^{(\omega)}_{\{\theta_{m},m\in
      P\}}(t)-\sum_{\ell=1}^{2s+1}\kappa_{\ell,t,\{\theta_{m},m\in
      P\}}(\infty)\varphi_\ell\;\bigg\|_{L^2}\;\le\;\frac{Ce^{K|t|}}{\omega^{1/2}}\; ,
  \end{equation}
\end{proposition}
\begin{proof}
  By direct computation one finds that the coefficients $\{\kappa_{j,t,\{\theta_{m},m\in P\}}(\omega)\}$
  satisfy the system of ODE
  \begin{equation}
    \left\{
      \begin{aligned}
	&i \partial_t \kappa_{j,t,\{\theta_{m},m\in P\}}(\omega)=\biggl\langle \;\varphi_{j}\;  ,\; V*\bigl\lvert\psi^{(\omega)}_{\{\theta_{m},m\in P\}}(t) \bigr\rvert_{}^2\bigl(\psi^{(\omega)}_{\{\theta_{m},m\in P\}}(t)\bigr)\; \biggr\rangle_{} \\
	&\kappa_{j,0,\{\theta_{m},m\in P\}}(\omega)=\left\{
          \begin{aligned}
            &e^{i\theta_{j}}\sqrt{\pi_{j}}& \text{if } j\in P\\
            &0& \text{if } j\notin P
          \end{aligned}
	\right .
      \end{aligned}
    \right .\quad .
  \end{equation}
  Let us compute
  \begin{equation*}
    \begin{split}
      \partial_t& \big|\kappa_{j,t,\{\theta_{m},m\in P\}}(\omega)-\kappa_{j,t,\{\theta_{m},m\in P\}}(\infty)\big|^2\\
      \;=\;&\mathrm{Im}\bigg[ \big( \overline{\kappa_{j,t,\{\theta_{m},m\in P\}}(\omega)}-\overline{\kappa_{j,t,\{\theta_{m},m\in P\}}(\infty)} \big)  \bigg( \biggl\langle \;\varphi_{j}\;  ,\; V*\bigl\lvert\psi^{(\omega)}_{\{\theta_{m},m\in P\}}(t) \bigr\rvert_{}^2\bigl(\psi^{(\omega)}_{\{\theta_{m},m\in P\}}(t)\bigr)\; \biggr\rangle_{}\\
      & - \biggl\langle \;\varphi_{j}\;  ,\; V*\biggl\lvert \sum_{\ell=1}^{2s+1}\kappa_{\ell,t,\{\theta_{m},m\in P\}}(\infty)\varphi_{\ell}  \biggr\rvert_{}^2\biggl(\sum_{\ell'=1}^{2s+1}\kappa_{\ell',t,\{\theta_{m},m\in P\}}(\infty)\varphi_{\ell'} \biggr)\; \biggr\rangle_{}  \bigg) \bigg]\\
      =\;&\mathrm{Im}\bigg[ \big( \overline{\kappa_{j,t,\{\theta_{m},m\in P\}}(\omega)}-\overline{\kappa_{j,t,\{\theta_{m},m\in P\}}(\infty)} \big) \\
      &\qquad\times \biggl\langle \;\varphi_{j}\;  ,\; V*\bigl\lvert\psi^{(\omega)}_{\{\theta_{m},m\in P\}}(t) \bigr\rvert_{}^2\biggl(\psi^{(\omega)}_{\{\theta_{m},m\in P\}}(t)-\sum_{\ell=1}^{2s+1}\kappa_{\ell,t,\{\theta_{m},m\in P\}}(\infty)\varphi_{\ell}\biggr)\; \biggr\rangle_{}\bigg]\\
      & +\mathrm{Im}\bigg[ \big(\overline{\kappa_{j,t,\{\theta_{m},m\in P\}}(\omega)}-\overline{\kappa_{j,t,\{\theta_{m},m\in P\}}(\infty)} \big) \\
      &\qquad\times \biggl\langle \;\varphi_{j}\;  ,\; V*\biggl(\,\biggl\lvert \sum_{\ell=1}^{2s+1}\kappa_{\ell,t,\{\theta_{m},m\in P\}}(\infty)\varphi_{\ell}  \biggr\rvert_{}^2- \bigl\lvert  \psi^{(\omega)}_{\{\theta_{m},m\in P\}}(t) \bigr\lvert^2\biggl)\\
      &\qquad\qquad\times \biggl(\sum_{\ell'=1}^{2s+1}\kappa_{\ell',t,\{\theta_{m},m\in P\}}(\infty)\varphi_{\ell'} \biggr)\; \biggr\rangle_{}  \bigg]\\
      =:\;&I_j+II_j\; .
    \end{split}
  \end{equation*}
  We treat the two terms separately. For $I_j$ we write, by H\"older inequality,
  \begin{equation*}
    \begin{split}
      |I_j|\le\;&\Big|\kappa_{j,t,\{\theta_{m},m\in P\}}(\omega)-\kappa_{j,t,\{\theta_{m},m\in P\}}(\infty)\Big| \;\bigg\| V*\bigl\lvert\psi^{(\omega)}_{\{\theta_{m},m\in P\}}(t) \bigr\rvert_{}^2 \varphi_j\bigg\|_2\;\\
      &\qquad\times \bigg\|\psi^{(\omega)}_{\{\theta_{m},m\in P\}}(t)-\sum_{\ell=1}^{2s+1}\kappa_{\ell,t,\{\theta_{m},m\in P\}}(\infty)\varphi_{\ell}\bigg\|_2\; .
    \end{split}
  \end{equation*}
  By the assumption \eqref{eq:bound_V1} on $V$ one finds
  \begin{equation*}
    \bigg\| V*\bigl\lvert\psi^{(\omega)}_{\{\theta_{m},m\in P\}}(t) \bigr\rvert_{}^2 \varphi_j\bigg\|_2\le \varepsilon \langle\psi^{(\omega)}_{\{\theta_{m},m\in P\}}(t),h_\omega\psi^{(\omega)}_{\{\theta_{m},m\in P\}}(t)\rangle+C_\varepsilon\le C\; ,
  \end{equation*}
  for some constant C that does not depend on $\omega$. Moreover,
  \begin{equation*}
    \psi^{(\omega)}_{\{\theta_{m},m\in P\}}(t)-\sum_{\ell=1}^{2s+1}\kappa_{\ell,t,\{\theta_{m},m\in P\}}(\infty)\varphi_{\ell}=\sum_{\ell=1}^{2s+1}\big(\kappa_{\ell,t,\{\theta_{m},m\in P\}}(\omega)-\kappa_{\ell,t,\{\theta_{m},m\in P\}}(\infty)\big)\varphi_\ell+\psi^{\perp}_{t}(\omega).
  \end{equation*}
  Hence,
  \begin{equation} \label{eq:finalIj}
    \begin{split}
      |I_j|\;\le&\; C\,\Big|\kappa_{j,t,\{\theta_{m},m\in P\}}(\omega)-\kappa_{j,t,\{\theta_{m},m\in P\}}(\infty)\Big|\, \\
      &\qquad\times\Big[\sum_{\ell=1}^{2s+1}\Big|\kappa_{\ell,t,\{\theta_{m},m\in P\}}(\omega)-\kappa_{\ell,t,\{\theta_{m},m\in P\}}(\infty)\Big|+ \|\psi_t^\perp(\omega)\|_2 \Big]\\
      \le&\;	C\,\Big|\kappa_{j,t,\{\theta_{m},m\in P\}}(\omega)-\kappa_{j,t,\{\theta_{m},m\in P\}}(\infty)\Big|\,\\
      &\qquad\times \Big[\sum_{\ell=1}^{2s+1}\Big|\kappa_{\ell,t,\{\theta_{m},m\in P\}}(\omega)-\kappa_{\ell,t,\{\theta_{m},m\in P\}}(\infty)\Big|+ \frac{C}{\omega^{1/2}} \Big]\; ,
    \end{split}
  \end{equation}
  having used Proposition \ref{prop:perp_small} in the last inequality.
	
  The estimate of $II_j$ is analogous, and yields
  \begin{equation*}
    \begin{split}
      |II_{j}|\;\le\;&C\,\Big|\kappa_{j,t,\{\theta_{m},m\in P\}}(\omega)-\kappa_{j,t,\{\theta_{m},m\in P\}}(\infty)\Big|\,\\
      & \Big[\sum_{\ell=1}^{2s+1}\Big|\kappa_{\ell,t,\{\theta_{m},m\in P\}}(\omega)-\kappa_{\ell,t,\{\theta_{m},m\in P\}}(\infty)\Big|+ \frac{C}{\omega^{1/2}} \Big]\; .
    \end{split}
  \end{equation*}
  Denoting $A(t):=\sum_{k=1}^{2s+1}|\kappa_{k,t\{\theta_{m},m\in P\}}(\omega)-\kappa_{k,t\{\theta_{m},m\in
    P\}}(\infty)\big|^2$, we have thus proven
  \begin{equation*}
    \dot A(t)\;\le\; C\big(A(t)+\omega^{-1}\big) \; ,
  \end{equation*}
  which implies, by Grönwall's inequality and using $A(0)=0$,
  \begin{equation}
    A(t)\;\le\;\frac{1}{\omega}e^{Ct}.
  \end{equation}
  This proves \eqref{eq:coefficient_convergence}. Eq. \eqref{eq:L_2_convergence} then follows immediately
  by comparing \eqref{eq:13} with \eqref{eq:small_perp} and \eqref{eq:coefficient_convergence}.
\end{proof}

We are now able to conclude the proof of \eqref{thm:convergence} for the limits taken in the order $N\to
\infty$, $\omega\to \infty$. Recall that, from \eqref{eq:11},
\begin{equation*}
  \gamma_{\infty,t}^{(1)}(\omega)=\frac{1}{(2\pi)^\Pi}\int_0^{2\pi} \bigl\lvert \psi^{(\omega)}_{\{\theta_{m},m\in P\}}(t)\bigr\rangle\bigl\langle \psi^{(\omega)}_{\{\theta_{m},m\in P\}}(t)\bigr\rvert \, \prod_{m\in P}^{} \mathrm{d}\theta_m\;.
\end{equation*}
We will show that the difference
\begin{equation}\label{eq:difference}
  \begin{split}
    D(\omega,t):=\;&\gamma^{(1)}_{\infty,t}(\omega)-\gamma^{(1)}_{\infty,t}(\infty)\\
    =\;&\gamma^{(1)}_{\infty,t}(\omega)-\sum_{i=1}^{2s+1}K_{j,t}(\infty)\lvert\varphi_{j}\rangle\langle\varphi_{j}\rvert\\
    &-\sum_{j<\ell=1}^{2s+1}\Big(K_{j\ell,t}(\infty)\lvert\varphi_{j}\rangle\langle\varphi_{\ell}\rvert-\overline{K_{j\ell,t}(\infty)}\lvert\varphi_{\ell}\rangle\langle\varphi_{j}\rvert\Big)
  \end{split}
\end{equation}
converges to zero in trace norm, as $\omega\to\infty$. To this aim, we introduce the identity
\begin{equation*}
  \begin{split}
    \sum_{j=1}^{2s+1}K_{j,t}(\infty)\lvert\varphi_{j}&\rangle\langle\varphi_{j}\rvert+\sum_{j<\ell=1}^{2s+1}\Big(K_{j\ell,t}(\infty)\lvert\varphi_{j}\rangle\langle\varphi_{\ell}\rvert+\overline{K_{j\ell,t}(\infty)}\lvert\varphi_{\ell}\rangle\langle\varphi_{j}\rvert\Big)\\
    &\;=\;\frac{1}{(2\pi)^\Pi}\int_0^{2\pi}\,\bigl\lvert \sum_{\ell=1}^{2s+1}\kappa_{\ell,t,\{\theta_{m},m\in P\}}(\infty)\varphi_\ell\bigr\rangle\bigl\langle \sum_{\ell'=1}^{2s+1}\kappa_{\ell',t,\{\theta_{m},m\in P\}}(\infty)\varphi_{\ell'}\bigr\rvert\,\prod_{m\in P}^{} \mathrm{d}\theta_m
  \end{split}
\end{equation*}
which is deduced by expanding products inside the integral. We write
\begin{equation*}
  \begin{split}
    \mathrm{Tr}&\big|D(\omega,t)\big|\\
    =\;&\mathrm{Tr}\bigg|\frac{1}{(2\pi)^\Pi}\int_{0}^{2\pi} \Bigl(\,\bigl\lvert \psi^{(\omega)}_{\{\theta_{m},m\in P\}}(t)\bigr\rangle\bigl\langle \psi^{(\omega)}_{\{\theta_{m},m\in P\}}(t)\bigr\rvert \\
    &\qquad\qquad\qquad\qquad-\bigl\lvert \sum_{\ell=1}^{2s+1}\kappa_{\ell,t,\{\theta_{m},m\in P\}}(\infty)\varphi_\ell\bigr\rangle\bigl\langle \sum_{\ell'=1}^{2s+1}\kappa_{\ell',t,\{\theta_{m},m\in P\}}(\infty)\varphi_{\ell'}\bigr\rvert  \,\Bigl)\, \prod_{m\in P}^{} \mathrm{d}\theta_m\,\bigg|\\
    \;\le\;&\frac{1}{(2\pi)^\Pi}\int_{0}^{2\pi}\mathrm{Tr}\bigg|\bigl\lvert \psi^{(\omega)}_{\{\theta_{m},m\in P\}}(t)\bigr\rangle\bigl\langle \psi^{(\omega)}_{\{\theta_{m},m\in P\}}(t)\bigr\rvert \\
    &\qquad\qquad\qquad\qquad-\bigl\lvert \sum_{\ell=1}^{2s+1}\kappa_{\ell,t,\{\theta_{m},m\in P\}}(\infty)\varphi_\jmath\bigr\rangle\bigl\langle \sum_{\ell'=1}^{2s+1}\kappa_{\ell',t,\{\theta_{m},m\in P\}}(\infty)\varphi_{\ell'}\bigr\rvert \;\bigg| \prod_{m\in P}^{} \mathrm{d}\theta_m\\\
    \;\le\;&\frac{2}{(2\pi)^\Pi}\int_{0}^{2\pi}\bigg\|\;\psi^{(\omega)}_{\{\theta_{m},m\in P\}}(t)-\sum_{\ell=1}^{2s+1}\kappa_{\ell,t,\{\theta_{m},m\in P\}}(\infty)\varphi_\ell\;\bigg\|_{2}\,\prod_{m\in P}^{} \mathrm{d}\theta_m,
  \end{split}
\end{equation*}
having used in the last step the bound $\mathrm{Tr}|p_{\psi_1}-p_{\psi_2}|\le2\|\psi_1-\psi_2\|_{L^2}$, where $p_{\psi_i}=\lvert
\psi_i\rangle\langle\psi_i\lvert$ is the orthogonal projection on the subspace spanned by $\psi_i$. By \eqref{eq:L_2_convergence} we
deduce
\begin{equation*}
  \mathrm{Tr}\big|D(\omega,t)\big|\;\le\;\frac{Ce^{K|t|}}{\omega^{1/2}}
\end{equation*}
which proves the trace-norm limit at fixed $t$
\begin{equation*}
  \lim_{\omega\to\infty}\lim_{N\to\infty}\gamma^{(1)}_{N,t}(\omega)=\gamma_{\infty,t}^{(1)}(\infty).
\end{equation*}

\section{Proof of Theorem \ref{thm:convergence}, case $\omega\to\infty$ followed by $N\to\infty$}
\label{sec:proof-crefthm:convergence2}

Let $P_N$ be the orthogonal projection whose range is
\begin{equation}
  \mathcal{H}_{P_N}\equiv\mathrm{ran} \,P_N:={\big(\operatorname{span}\{\varphi_1,\dots,\varphi_{2s+1}\}\big)}^{\bigvee N},
\end{equation}
that is, the space spanned by many-body vectors in which all particles occupy one of the one-body orbitals
$\varphi_k$'s, $k=1,\dots,2s+1$. Let us define
\begin{equation}
  \Phi_\omega(t):=P_N\Psi(t)\equiv P_Ne^{-itH_{\omega,N}}\Psi_0
\end{equation}
and
\begin{equation}
  \Phi_\infty(t):=e^{-it P_NV_NP_N}\Psi_0\; ,
\end{equation}
having used the notation
\begin{equation*}
  V_N:=\frac{1}{N}\sum_{j<k}V(x_j-x_k)\; .
\end{equation*}
It should be kept in mind that the vector $\Psi(t)$ depends on $\omega$, even if though we do not mention
it explicitly.  We will prove that $\Phi_\infty(t)$ is the limit of $\Psi(t)$ as $\omega\to\infty$. To
this aim it is sufficient to prove that $(1-P_N)\Psi(t)$ vanishes in the norm of $\mathcal{H}_N$ as
$\omega\to\infty$ (see Proposition \ref{prop:many_body_perp} below), and that $\Phi_\omega(t)$ converges
to $\Phi_\infty(t)$ (Proposition \ref{prop:many_body_omega} below).

\begin{proposition}[\textbf{$N$-body restriction to ground states of $h_\omega$}]\mbox{}\\
  \label{prop:many_body_perp}
  In the same assumptions of Theorem \ref{thm:convergence}, let $\Psi(t)=e^{-itH_{\omega,N}}\Psi_0$, with
  initial datum
  \begin{equation*}
    \Psi_0=\varphi_1^{\otimes f_1(N)}\vee\dots\vee \varphi_{2s+1}^{\otimes f_{2s+1}(N)}.
  \end{equation*}
  Then, for any $t\in \mathbb{R}$, there exists $C>0$, independent of $\omega$ and $t$, such that
  \begin{equation} \label{eq:small_excited_manybody}
    \big\|(1-P_N)\Psi(t)\big\|_{\mathcal{H}_N}^2\le\frac{C N}{\omega}\; .
  \end{equation}
\end{proposition}

\begin{proof}
  On the one hand, by conservation of the many-body energy we have
  \begin{equation}
    \label{eq:manybodybound}
    \langle\Psi(t),H_{\omega,N}\Psi(t)\rangle=\langle\Psi_0, H_{\omega,N}\Psi_0\rangle=\frac{N-1}{2}\langle\Psi_0,V(x_1-x_2)\Psi_0\rangle\; .
  \end{equation}
	
  On the other hand, by the relative boundedness of $V$ with respect to $h_\omega$, for every
  $\varepsilon>0$ we have
  \begin{equation*}
    \langle\Psi(t),H_{\omega,N}\Psi(t)\rangle \ge (1-\varepsilon)\sum_{j=1}^N\langle\Psi(t),h_{\omega,j}\Psi(t)\rangle-C_\varepsilon\;.
  \end{equation*}
  Since every $h_{\omega,j}$ vanishes on the range of $P_N$ we have
  \begin{equation*}
    \begin{split}
      \langle\Psi(t),H_{\omega,N}\Psi(t)\rangle\ge\;& (1-\varepsilon)\,\sum_{j=1}^{N}\langle(1-P_N)\Psi(t),h_{\omega,j}(1-P_N)\Psi(t)\rangle-C_\varepsilon
      \\\ge\;&\omega(1-\varepsilon)\|(1-P_N)\Psi(t)\|_{\mathcal{H}_N}^2-C_\varepsilon\; ,
    \end{split}
  \end{equation*}
  having used the gap of $h_\omega$ in the last step. Comparing this with \eqref{eq:manybodybound}
  concludes the proof.
\end{proof}

\begin{proposition}[\textbf{Convergence to limit $N$-body dynamics}]\mbox{}\\
  \label{prop:many_body_omega}
  In the same assumptions of Theorem \ref{thm:convergence}, for every $t\in \mathbb{R}$ fixed we have
  \begin{equation}
    \lim_{\omega\to\infty}\,\Bigl\|\Phi_\omega(t)-\Phi_\infty(t)\Bigr\|_{\mathcal{H}_{P_N}}=0\; .
  \end{equation}
\end{proposition}

\begin{proof}
  Using the fact that $P_NH_{\omega,N}(1-P_N)=P_NV_N(1-P_N)$, we deduce
  \begin{equation*}
    i\partial_t\Phi_\omega(t)=P_NV_NP_N\Phi_\omega(t)+P_NV_N(1-P_N)\Psi(t)\; ,
  \end{equation*}
  which is equivalent to the integral equation
  \begin{equation*}
    \Phi_\omega(t)=e^{-itP_NV_NP_N}\Psi_0+\int_0^t\mathrm{d}\tau\,e^{-i(t-\tau)P_NV_NP_N}P_NV_N(1-P_N)\Psi(s)\; .
  \end{equation*}
  Comparing with the definition of $\Phi_\infty(t)$ we obtain, using \eqref{eq:bound_V2} with $\omega=1$ to estimate $V(x_j-x_k)$,
  \begin{equation*}
    \Bigl\|\Phi_\omega(t)-\Phi_\infty(t)\Bigr\|_{\mathcal{H}_{P_N}}\le \varepsilon\lvert t  \rvert_{}^{}\sup_{\tau\in [0,\lvert t  \rvert_{}^{}]}\Bigl\|\sum_{j=1}^{N}h_{1,j}(1-P_N)\Psi(s)\Bigr\|_{\mathcal{H}_N}+C_{\varepsilon,N}\lvert t  \rvert_{}^{}\sup_{\tau\in [0,\lvert t  \rvert_{}^{}]} \Bigl\lVert (1-P_N)\Psi(s)  \Bigr\rVert_{\mathcal{H}_N}^{}\; .
  \end{equation*}
  Now, since both $(1-P_N)\Psi(\tau)$ and $\sum_{j=1}^Nh_{1,j}(1-P_N)\Psi(\tau)$ are strongly
  continuous with respect to $\tau\in [0,\lvert t \rvert_{}^{}]$, and since $\sum_{j=1}^Nh_{1,j}$ is
  a closed operator, the result follows from \eqref{eq:small_excited_manybody}.
\end{proof}

From Proposition \ref{prop:many_body_perp} and Proposition \ref{prop:many_body_omega} we deduce the $L^2$-many-body bound
\begin{equation}
  \lim_{\omega\to\infty}\|\Psi_t-\Phi_\infty(t)\|_{\mathcal{H}_N}\equiv\lim_{\omega\to\infty}\Big\|e^{-itH_{\omega,N}}\Psi_0-e^{-itP_NV_NP_N}\Psi_0\Big\|_{\mathcal{H}_N}=0\;.
\end{equation}
Let $\gamma^{(1)}_{N,t}(\infty)$ be the 1-RDM associated to $\Phi_{\infty}(t)$. Then the above limit
implies in particular the trace-norm limit
\begin{equation*}
  \lim_{\omega\to\infty}\gamma^{(1)}_{N,t}(\omega)=\gamma^{(1)}_{N,t}(\infty).
\end{equation*}

To complete the proof of Theorem \ref{thm:convergence} it remains to prove that
$\gamma^{(1)}_{N,t}(\infty)$ converges to $\gamma_{\infty,t}^{(1)}(\infty)$ in trace norm as $N\to
\infty$. This is in fact a problem of finite-dimensional semiclassical analysis (with $N^{-1}$ as
semiclassical parameter). Consider the operator
\begin{equation*}
  W_N:=\frac{1}{N}\bigoplus _{n\in \mathbb{N}}P_N\sum_{1=j<k}^{n}V(x_j-x_k)P_N
\end{equation*}
on the Fock space $\Gamma_{\mathrm{s}}(\mathbb{C}^{2s+1})\cong L^2 (\mathbb{R}^{2s+1} )$ (where the isomorphism is intended between unitarily equivalent irreducible
representations of the C*-algebra of canonical commutation relations). On one hand it agrees with
$P_NV_NP_N$ when restricted to the sector with $n=N$, and on the other hand it is the Wick quantization of
a symbol $\sigma(\zeta,\bar{\zeta})$ on $\mathbb{C}^{2s+1}$. Such symbol $\sigma$ is, if we make the
identifications $\zeta=(\kappa_1,\dotsc,\kappa_{2s+1})$, and $u_{\zeta}=\sum_{\jmath
  =1}^{2s+1}\kappa_{\jmath}\varphi_{\jmath}\in \mathcal{H}$,
\begin{equation}
  \label{eq:15}
  \sigma(\zeta,\bar{\zeta})=\frac{1}{2}\int_{\mathbb{R}^d}^{}V(x-y)\bar{u}_{\zeta}(x)\bar{u}_{\zeta}(y)u_{\zeta}(x)u_{\zeta}(y)  \mathrm{d}x\mathrm{d}y\; .
\end{equation}
Eq.\ \eqref{eq:15} precisely defines the energy of the Hamilton-Jacobi equations \eqref{eq:limit_eq}. The
trace-norm convergence of reduced density matrices is well-known in this finite-dimensional context, and yields the sought result. This
concludes the proof of \eqref{thm:convergence}.

\section{Proof of Proposition \ref{prop:matrices}}
\label{sec:proof-crefthm:matrices}

To prove Proposition \ref{prop:matrices} it is important to better understand the combinatorial factors
appearing in $\gamma_{N,\mathrm{frag}}$ and $ \gamma_{\infty,\mathrm{frag}} ^{(p)} $. We can rewrite the wave function $
\varphi _1 ^{\otimes f_1(N)} \vee \dots \vee \varphi _{2s+1} ^{\otimes f_{2s+1}(N)} $ as
\begin{equation}
  \begin{split}
    \varphi _1 ^{\otimes f_1(N)} \vee \dotsm \vee \varphi _{2s+1} ^{\otimes f_{2s+1}(N)}
    \left(
      x _1,\ldots
      x _N
    \right):
    &=C _{\{f(N)\}}
    \sum _{\sigma \in \mathscr{S} _N}
    \prod _{n=1} ^N
    \sum _{j=1} ^{2s+1}
    \chi _{F _j} (n)
    \varphi _j (x_{\sigma(n)})
    \\
    &= C _{\{f(N)\}}
    \sum _{\sigma \in \mathscr{S} _N}
    \prod _{n=1} ^N
    \sum _{j=1} ^{2s+1}
    \chi _{\sigma(F _j)} (n)
    \varphi _j (x_n)\; ,
  \end{split}
\end{equation}
where the normalization constant is
\begin{equation*}
  C _{\{f(N)\}} = \frac{1}{{\bigl( N! \prod _{j =1} ^{2s+1} f _j !  \bigr)}^{\frac{1}{2}}},
\end{equation*}
and
\begin{equation*}
  F _j = \Bigl( \sum
  _{l=1} ^{j-1} f _l(N), \sum _{l=1} ^j f _l(N) \Bigr] \cap \mathbb{N}\quad\text{for } j \in \bigl \{ 1, \ldots, 2s+1 \bigr \}.
\end{equation*}
It follows from the orthogonality of the vectors $\varphi_j$ that cancellations occur when performing the
partial trace of $\gamma_{N,\mathrm{frag}}$; nonetheless the number of vectors $ \varphi_j $ is the same
on each side of the projection. More precisely, the integral kernel of the $p$-RDM is
\begin{equation}
  \begin{split}
    \label{eq:trace_f_calc}
    \gamma^{(p)}_{N,\mathrm{frag}}&(x_1,\dots,x_p;y_1,\dots,y_p)\\
    =\;&\big[\mathrm{Tr}_{\mathcal{H}_{N-p}} (\gamma _{N,\mathrm{frag}})\big]  \left( x _1, \ldots, x _p; y _1, \ldots, y _p \right) \\
    =\;& C _{\{f(N)\}} ^2 \sum
    _{\sigma, \tau \in \mathscr{S} _N} \left( \prod _{n=1} ^p \sum _{j=1} ^{2s+1} \chi _{\sigma(F _j)\cap\left\{1,\ldots,k\right\}} (n) \varphi _j
      (x_n) \right)
    \\&
    \times \left( \prod _{\nu=1} ^p \sum _{\gamma=1} ^{2s+1} \chi _{\tau(F _\gamma)\cap\left\{1,\ldots,k\right\}} (\nu) \overline{\varphi _\gamma (y_{\nu})}
    \right) \left( \prod _{m=k+1} ^N \sum _{l=1} ^{2s+1} \chi _{\sigma(F _l)\cap\tau(F _l)\cap\left\{k+1,\ldots,N\right\}} (m) \right).
  \end{split}
\end{equation}
The last term in the product naturally translate on a condition on the family of possible permutations
that we can choose
from. 
Let us identify, for any permutation $ \sigma \in \mathscr{S} _N $, the set $ I[\sigma] = \sigma ^{-1}
\bigl(\bigl\{ 1, \ldots, p \bigr\}\bigr) $. The permutation $ \sigma $ can be seen as a permutation from $
I[\sigma] $ to $ p $ elements, and another from the complement of $ I[\sigma] $ to (the other) $ N-p $
elements. This idea can be carried further, and it is easy to realize that every permutation $ \sigma \in
\mathscr{S} _N $ can be thought of as a triple $ \bigl( \sigma _1, \sigma _2, I \bigr) $ with $ \sigma _1
\in \mathscr{S} _p $, $ \sigma _2 \in \mathscr{S} _{N-p} $ and $ I \subseteq \bigl\{ 1, \ldots, N \bigr\}
$ with $ \# I = p $.

Now it is important to understand how the choice of $ p $ particles affects the total distribution of
particles. In this sense, one should first notice that $ f(N) $ (as a frequency distribution) is
identified uniquely by the partition of connected subsets $ F=\bigl\{ F _j,\ 1\le j \le 2s+1 \bigr\}
$. Both the family $ F \cap I[\sigma] $ and $ F \cap \bigl( \bigl\{1, \ldots, N \bigr\} \smallsetminus
I[\sigma] \bigr) $ are partitions of connected subsets, respectively of $ I [\sigma] $ and of its
complement.  To these two families we can now associate two different frequency distributions $ g $ and $
h $, respectively in the sets $ \mathcal{F} _{p,2s+1} ^N $ and $ \mathcal{F} _{N-p,2s+1} ^N $, where
\begin{equation}
  \mathcal{F} _{p,\alpha} ^N:
  =\biggl\{
  g: \left\{
    1,
    \ldots,
    \alpha
  \right\}
  \to \mathbb{N}:\
  g _j \le f _j(N),\
  \sum _{j=1} ^\alpha
  g _j =k
  \biggr\}.
\end{equation}
Denote $ F ^g := F \cap I[\sigma] $, and $ F ^h := F \cap \bigl( \bigl\{1, \ldots, N \bigr\}
\smallsetminus I[\sigma] \bigr) $. These two families of sets (and the corresponding distributions $ g $
and $ h $) are identified, in a non-unique fashion, by the set $ I[\sigma] $. It is hence possible to
rewrite \eqref{eq:trace_f_calc} as
\begin{equation}
  \begin{split}
    \gamma^{(p)}_{N,\mathrm{frag}}&(x_1,\dots,x_p;y_1,\dots,y_p)\\
    =\;&\sum _{\substack{ \sigma _1, \tau _1 \in \mathscr{S} _p,\
        \sigma _2, \tau _2 \in \mathscr{S} _{N-p} \\
        I, J \subseteq\bigl\{1,\ldots,N\bigr\},
        \#I=\#J=p \\
        \sigma _2 \bigl(F _l ^{h(I)}\bigr) =\tau _2 \bigl(F _l ^{h(J)}\bigr), 1\le l \le 2s+1 }} \mspace{-25mu}C _{\{f(N)\}} ^2\\
    &\quad\times \Big( \prod _{n=1} ^p \sum _{l=1} ^{2s+1} \chi _{\sigma _1(F _l ^{g(I)})} (n) \varphi _l (x_n) \Big) \Big( \prod
    _{\nu=1} ^p \sum _{\gamma=1} ^{2s+1} \chi _{\tau _1(F _\gamma ^{g(J)})} (\nu) \overline{\varphi _\gamma (y_{\nu})} \Big)\; .
  \end{split}
\end{equation}

To get rid of $ \sigma _2 $ and $ \tau _2 $ it is useful to compute the number of permutations that
respect the constraint in the sum, that is
\begin{multline}
  \#\left\{ \sigma _2, \tau _2 \in \mathscr{S} _{N-p}:\ \sigma _2\left( F _l ^{h(I)} \right) =\sigma
    _2\left( F _l ^{h(J)} \right),\ 1\le l \le 2s+1 \right\}=
  \\
  =\left\{
    \begin{array}{ll}
      \left(N-p\right)!
      \prod _{l=1} ^{2s+1}
      \left[
        h(I) _l
      \right]!
      \equiv C _{h(I), N-p} ^{-2}
      & \mbox{if } g(I)=g(J);\\
      0 & \mbox{else.}
    \end{array}
  \right.
\end{multline}
Moreover, since the terms of the summation do not depend directly on $ I $ and $ J $, we would like to
substitute them with the possible choices of $ g $. To do so, it is necessary to understand how many
choices of $ I $ correspond to the same distribution function $ g $. Equivalently, we should know in how
many ways it is possible to choose $ g _j $ elements in $ f _j(N) $. Clearly, the answer is given by a
binomial coefficient. Therefore, it follows that
\begin{equation}
  \begin{split}
    \gamma^{(p)}_{N,\mathrm{frag}}&(x_1,\dots,x_p;y_1,\dots,y_p)
    \\
    &=\sum _{\substack{
        \sigma, \tau \in \mathscr{S} _p \\
        g \in \mathcal{F} _{p,2s+1} ^N
      }}
    \frac{C _{\{f(N)\}} ^2}{C _{f(N)-g,N-p} ^2}
    \left[
      \prod _{j=1} ^{2s+1}
      \binom{f _j (N)}{g _j}
    \right] ^2\\
    &\quad\times
    \left(
      \prod _{n=1} ^p
      \sum _{l=1} ^{2s+1}
      \chi _{\sigma (F _l ^g)} (n)
      \varphi _l (x_n)
    \right)
    \left(
      \prod _{\nu=1} ^p
      \sum _{\gamma=1} ^{2s+1}
      \chi _{\tau (F _\gamma ^g)} (\nu)
      \overline{\varphi _\gamma (y_{\nu})}
    \right).
  \end{split}
\end{equation}

It is not difficult to see the operator corresponding to that kernel is a combination of projectors:
\begin{equation}
  \gamma^{(p)}_{N,\mathrm{frag}}
  =\sum _{g \in \mathcal{F} _{p,2s+1} ^N}
  \binom{N}{p} ^{-1}
  \left[
    \prod _{j=1} ^{2s+1}
    \binom{f _j (N)}{g _j}
  \right]
  \bigl\lvert
  \varphi_1^{\otimes g _1}
  \vee \dotsm
  \vee \varphi_{2s+1}^{\otimes g _{2s+1}}
  \bigr\rangle
  \bigl\langle
  \varphi_1^{\otimes g _1}
  \vee \dotsm
  \vee \varphi_{2s+1}^{\otimes g _{2s+1}}
  \bigr\rvert\; .
\end{equation}

To compute the limit $N\to \infty$ it is convenient to calculate the limit of each coefficient in the
sum. Recall that by Assumption \ref{assum:one_body} $ f _j = \pi _j N +o(N) $. In particular this implies
that for any strictly positive integer $ \lambda $ smaller than $ f _j $ we have $ \binom{f _j}{\lambda} =
\frac{1}{\lambda !}  N ^\lambda \bigl( \pi _j ^\lambda +o(1) \bigr) $. Hence we get
\begin{equation}
  \begin{split}
    \label{eq:trace_f_coeff_comput}
    \lim _{N\to+\infty}
    \binom{N}{p} ^{-1}
    \prod _{j=1} ^{2s+1}
    \binom{f _j (N)}{g _j}
    =\lim _{N\to+\infty}
    \binom{N}{p} ^{-1}
    \prod _{\substack{
        j\in \left\{
          1,
          \ldots,
          2s+1
        \right\}\\
        j:\ g _j\neq 0
      }}
    \binom{f _j (N)}{g _j}
    =k!
    \prod _{\substack{
        j\in \left\{
          1,
          \ldots,
          2s+1
        \right\}\\
        j:\ g _j\neq 0
      }}\frac{\pi _j ^{g _j}}{g _j!}\; .
  \end{split}
\end{equation}
From \eqref{eq:trace_f_coeff_comput} it follows that
\begin{equation}
  \begin{split}
    \gamma^{(p)}_{\infty,\mathrm{frag}}:&= \underset{N\to \infty}{\mathrm{w^{*}-} \lim}\,
    \gamma^{(p)}_{N,\mathrm{frag}}
    \\
    &=\sum _{g \in \mathcal{F} _{p,2s+1}}
    p!\left[
      \prod _{\substack{
          j\in \left\{
            1,
            \ldots,
            2s+1
          \right\}\\
          j:\ g _j\neq 0
        }}\frac{\pi _j ^{g _j}}{g _j!}
    \right]
    \bigl\lvert
    \varphi_1^{\otimes g _1}
    \vee \dotsm
    \vee \varphi_{2s+1}^{\otimes g _{2s+1}}
    \bigr\rangle
    \bigl\langle
    \varphi_1^{\otimes g _1}
    \vee \dotsm
    \vee \varphi_{2s+1}^{\otimes g _{2s+1}}
    \bigr\rvert\; .
  \end{split}
\end{equation}
The rank of the effective $p$-RDM is thus the following:
\begin{equation}
  \operatorname{Rank} \; \gamma_{\infty,\mathrm{frag}}^{(p)}=
  \# \Bigl\{
  g \in
  \mathcal{F} _p:\
  \pi _j=0
  \Rightarrow g _j=0,
  1\le j \le s
  \Bigr\}
  =\binom{p+\Pi-1}{p}\; .
\end{equation}

\appendix

\section{Systems with many bosons and semiclassical analysis}
\label{sec:mean-field-descr}

In order to better understand the main mathematical tools used throughout this paper, let us recast the
large $N$ approximation as a semiclassical problem. Let us start with a simple remark: \emph{it is always
  possible to see an $N$-particle bosonic vector $\psi\in \mathcal{H}_N=\bigvee_{j=1}^N \mathcal{H}$, where
  $\mathcal{H}$ is a separable ``one-boson'' Hilbert space (in the main body of the paper it is
  $L^2(\mathbb{R}^d,\mathbb{C}^{2s+1})$), as the only non-zero component of a vector $\Psi=(0,\dotsc,0,\psi,0,\dotsc)$ in the symmetric Fock
  space $\Gamma_{\mathrm{s}}(\mathcal{H})=\bigoplus_{n\in \mathbb{N}}\mathcal{H}_n $.}  Therefore it is possible to interpret any
$N$-particle bosonic density matrix, with $N$ fixed, as a density matrix with an arbitrary number of
identical particles and probability one of having exactly $N$ particles. It is also possible to interpret
an $N$-body Hamiltonian $H_N$ (with pair interactions) as an Hamiltonian $H$ on the Fock space that agrees
with $H_N$ \emph{on the $N$-particles sector} and that commutes with the number operator: let
\begin{equation*}
  H_N=\sum_{j=1}^Nh_j+ \frac{1}{N}\sum_{j<k}^{}V(x_j-x_k)
\end{equation*}
be the self-adjoint Hamiltonian defined by \eqref{eq:mf_Hamiltonian} (where we have omitted the $\omega$
dependence for simplicity); then the Hamiltonian $H$ on $\Gamma_{\mathrm{s}}(\mathcal{H})$ defined by
\begin{equation} \label{eq:H_fock}
  H=\int_{}^{}h(x,y)a^{*}(x)a(y)  \mathrm{d}x\mathrm{d}y +\frac{1}{2N}\int_{}^{}V(x-y)a^{*}(x)a^{*}(y)a(x)a(y)  \mathrm{d}x \mathrm{d}y\;, 
\end{equation}
where $h(x,y)$ is the integral kernel of the self-adjoint operator $h$, and $a^{*}$, $a$ are the bosonic
creation and annihilation operator-valued distributions, is self-adjoint (see, \emph{e.g.},
\cite{ginibre1970cmpren,falconi2015mpag}) and agrees with $H_N$ when restricted to $\mathcal{H}_N$ (so in
particular $H\Psi\bigr\rvert_{\mathcal{H}_N}=H_N\psi$). So for any vector $\psi\in \mathcal{H}_N$, we have $e^{-it
  H_N}\psi=e^{-it H}\Psi\bigr\rvert_{\mathcal{H}_N}$, the other components of the latter being zero. In other words,
we can study the evolution of an $N$-body system with pair interactions directly in Fock space
(restricting to vectors whose only non-zero component is in the $N$-particle sector).

The limit $N\to \infty$ can be actually recast as a classical limit. To see that, define the semiclassical parameter $\varepsilon=N^{-1}$ and observe that $H$ in \eqref{eq:H_fock} depends on $\varepsilon$, as well as the aforementioned Fock space vector $\Psi$ does. From now on, we will emphasize the $\varepsilon$-dependence by adding a subscript $_\varepsilon$. Let us also define rescaled creation and annihilation operator-valued distributions
$a^{*}_{\varepsilon}=\sqrt{\varepsilon}a^{*}$ and $a_{\varepsilon}=\sqrt{\varepsilon}a$ satisfying
\begin{equation*}
  [a_{\varepsilon}(x),a^{*}_{\varepsilon}(y)]=\varepsilon\delta(x-y)\; .
\end{equation*}
If we rewrite $H$ in terms of $a^{*}_{\varepsilon}$ and $a_{\varepsilon}$, we obtain
$H=\frac{1}{\varepsilon}H_{\varepsilon}$, where
\begin{equation*}
  H_{\varepsilon}=\int_{}^{}h(x,y)a^{*}_{\varepsilon}(x)a^{}_{\varepsilon}(y)  \mathrm{d}x\mathrm{d}y + \frac{1}{2}\int_{}^{}V(x-y)a^{*}_{\varepsilon}(x)a^{*}_{\varepsilon}(y)a_{\varepsilon}(x)a_{\varepsilon}(y)  \mathrm{d}x\mathrm{d}y
\end{equation*}
coincides with $\frac{H}{N}$, that is, the energy-per-particle operator. Therefore, the evolution of the
system is described by $e^{-i \frac{t}{\varepsilon}H_{\varepsilon}}\Psi_{\varepsilon}$, and the creation and annihilation operators
corresponding to the canonical field observables of the system satisfy ``semiclassical'' $\varepsilon$-dependent
commutation relations. In other words, the parameter $\varepsilon$ for this theory is perfectly analogous to $\hslash$ in
ordinary quantum mechanics, and the system admits therefore a semiclassical description.

The main difference with respect to usual semiclassical analysis in quantum mechanics is that in this case
the phase-space is infinite-dimensional: in fact, it can be identified with the one-boson Hilbert space
$\mathcal{H}$. This makes our theory a special case of semiclassical bosonic quantum field theories. The
latter have been, in recent years, an active field of study (see, \emph{e.g.},
\cite{ammari2008ahp,falconi2013jmp,ammari2014jsp,ammari2015asns,amour2015jfa,falconi2016arxiv,amour2016alnarxiv,ammari2017arxiv,ammari2017sima,falconi2017ccm,correggi2017ahp,correggi2017arxiv}). In
the present paper we take advantage of the techniques developed in the above papers, see in particular
Theorem \ref{thm:mf_fragm}.

In semiclassical analysis, one sees quantum observables as the quantization of classical symbols,
\emph{i.e.}\ real-valued functions on the phase space. As we already mentioned, for mean field theories
the phase space can be identified with the one-boson space $\mathcal{H}$. The quantum Hamiltonian $H_{\varepsilon}$
is the quantization of the energy functional
\begin{equation}
  \label{eq:4}
  \mathcal{E}(u)=\int_{}^{}h(x,y)\bar{u}(x)u(y)  \mathrm{d}x \mathrm{d}y + \frac{1}{2}\int_{}^{}V(x-y)\lvert u(x)  \rvert_{}^2\lvert u(y)  \rvert_{}^2  \mathrm{d}x\mathrm{d}y\; ,
\end{equation}
that is densely defined on $D(h)\subseteq \mathcal{H}$, if $h$ and $V$ satisfy
Assumptions~\ref{assum:h},~\ref{assum:V} respectively. The quantization procedure used to pass from
$\mathcal{E}(u)$ to $H_{\varepsilon}$ is the so-called Wick quantization. Roughly speaking, one should replace every
$u(x)$ with $a_{\varepsilon}(x)$, and $\bar{u}(x)$ with $a^{*}_{\varepsilon}(x)$, and then put all the creation operators
$a^{*}_{\varepsilon}$ on the left of the annihilation operators $a_{\varepsilon}$ (for a detailed account of Wick quantization
in quantum field theories, see \emph{e.g.}\ \cite{ammari2008ahp}).

With the tools of semiclassical analysis, it is possible to investigate the limit behavior $\varepsilon\to 0$,
\emph{i.e.} the mean field limit $N\to \infty$ of the $N$-body theory. The semiclassical counterpart of a
$N$-body quantum state $\gamma_{\varepsilon}$, that is, a positive trace-class operator of trace one on $\mathcal{H}_N$,
is a probability measure on $\mathcal{H}$, called a \emph{Wigner measure}. In particular, given any such
$(\gamma_{\varepsilon})_{\varepsilon\in (0,1)}$, there exists a subsequence $\varepsilon_n\to 0$ and a unique Radon Borel probability measure $\mu$
(possibly depending on the chosen subsequence), concentrated on the unit ball of $\mathcal{H}$, such that
in weak-* topology
\begin{equation}
  \label{eq:5}
  \gamma_{\infty}^{(p)}=\lim_{n\to \infty}\gamma_{\varepsilon_n}^{(p)}=\int_{\mathcal{H}}^{}\,\bigl\lvert \,\underbrace{u\otimes \dotsm\otimes u}_{p} \,\bigr\rangle\bigl\langle\, \underbrace{u\otimes \dotsm\otimes u}_{p} \,\bigr\rvert  \,\mathrm{d}\mu(u)\; ,
\end{equation}
where $\gamma_{\varepsilon}^{(p)}$ are the $p$-RDMs of the density matrix $\gamma_{\varepsilon}$. Such type of result, known as quantum
de Finetti theorem, has been proved in
\cite{stormer1969jfa,Hudson1976,ammari2011jmpa,lewin2014am}. Therefore, the effective $p$-RDM is always an
average of rank one projections of a special form, in which one projects onto the $p$-th tensor product of
a single one-boson state. It is the measure $\mu$ that induces correlations in the $p$-RDMs: if the measure
is concentrated on a point $\tilde{u}$ (or on a circle $\{e^{i\theta}\tilde{u},\theta\in [0,2\pi]\}$), there is no correlation in
the sense that the resulting $p$-RDM is of the form
\begin{equation*}
  \bigl\lvert \,\underbrace{\tilde{u}\otimes \dotsm\otimes \tilde{u}}_{p} \,\bigr\rangle\bigl\langle\, \underbrace{\tilde{u}\otimes \dotsm\otimes \tilde{u}}_{p} \,\bigr\rvert\; .
\end{equation*}

Semiclassical techniques also allow to characterize explicitly the time evolution in the mean-field
limit. Let $\gamma_{\varepsilon}(t)=e^{-i\frac{t}{\varepsilon}H_{\varepsilon}}\gamma_{\varepsilon}e^{i\frac{t}{\varepsilon}H_{\varepsilon}}$ be the many-body evolution of the
density matrix $\gamma_{\varepsilon}$. Then it follows that (see, \emph{e.g.}, \cite{ammari2015asns,ammari2016cms})
\begin{equation}
  \label{eq:8}
  \begin{split}
    \gamma_{\infty}^{(p)}(t)=\lim_{n\to \infty}\gamma_{\varepsilon_n}^{(p)}(t)&=\int_{\mathcal{H}}^{}\,\bigl\lvert \,\underbrace{u\otimes \dotsm\otimes u}_{p} \,\bigr\rangle\bigl\langle\, \underbrace{u\otimes \dotsm\otimes u}_{p} \,\bigr\rvert  \,\mathrm{d}\mu_t(u)\\
    &=\int_{\mathcal{H}}^{}\,\bigl\lvert \,\underbrace{u_t\otimes \dotsm\otimes u_t}_{p} \,\bigr\rangle\bigl\langle\, \underbrace{u_t\otimes \dotsm\otimes u_t}_{p} \,\bigr\rvert  \,\mathrm{d}\mu(u)\; ,
  \end{split}
\end{equation}
where $u_t$ is the solution of the nonlinear Hartree Cauchy problem
\begin{equation}
  \label{eq:17}
  \begin{cases}
    &i\partial_tu_t=hu_t+V*\lvert u_t \rvert_{}^2u_t\\
    &u_0=u
  \end{cases}\quad .
\end{equation}
In other words, the Wigner measure $\mu_t$ at any time $t\in \mathbb{R}$ is the pushforward by the nonlinear Hartree
flow of the initial Wigner measure $\mu$.


\begin{thebibliography}{MHUB06}

\bibitem[ABN19]{ammari2017arxiv}
Zied Ammari, S\'ebastien Breteaux, and Francis Nier.
\newblock Quantum mean field asymptotics and multiscale analysis.
\newblock {\em Tunis. J. Math.}, 1(2):221--272, 2019.

\bibitem[AF14]{ammari2014jsp}
Zied Ammari and Marco Falconi.
\newblock {W}igner measures approach to the classical limit of the {N}elson
  model: {C}onvergence of dynamics and ground state energy.
\newblock {\em J. Stat. Phys.}, 157(2):330--362, 2014.

\bibitem[AF17]{ammari2017sima}
Zied Ammari and Marco Falconi.
\newblock {Bohr's correspondence principle for the renormalized Nelson model}.
\newblock {\em SIAM J. Math. Anal.}, 49(6):5031--5095, 2017.

\bibitem[AFP16]{ammari2016cms}
Zied Ammari, Marco Falconi, and Boris Pawilowski.
\newblock On the rate of convergence for the mean field approximation of
  bosonic many-body quantum dynamics.
\newblock {\em Commun. Math. Sci.}, 14(5):1417--1442, 2016.

\bibitem[AJN15]{amour2015jfa}
Laurent Amour, Lisette Jager, and Jean Nourrigat.
\newblock On bounded pseudodifferential operators in {W}iener spaces.
\newblock {\em J. Funct. Anal.}, 269(9):2747--2812, 2015.

\bibitem[ALN19]{amour2016alnarxiv}
Laurent Amour, Richard Lascar, and Jean Nourrigat.
\newblock {Weyl calculus in Wiener spaces and in QED}.
\newblock {\em J. Pseudo-Differ. Oper. Appl.}, 10(1):1--47, 2019.

\bibitem[AN08]{ammari2008ahp}
Zied Ammari and Francis Nier.
\newblock Mean field limit for bosons and infinite dimensional phase-space
  analysis.
\newblock {\em Ann. Henri Poincar\'e}, 9(8):1503--1574, 2008.

\bibitem[AN09]{ammari2009jmp}
Zied Ammari and Francis Nier.
\newblock Mean field limit for bosons and propagation of {W}igner measures.
\newblock {\em J. Math. Phys.}, 50(4):042107, 16, 2009.

\bibitem[AN11]{ammari2011jmpa}
Zied Ammari and Francis Nier.
\newblock Mean field propagation of {W}igner measures and {BBGKY} hierarchies
  for general bosonic states.
\newblock {\em J. Math. Pures Appl. (9)}, 95(6):585--626, 2011.

\bibitem[AN15]{ammari2015asns}
Zied Ammari and Francis Nier.
\newblock Mean field propagation of infinite-dimensional {W}igner measures with
  a singular two-body interaction potential.
\newblock {\em Ann. Sc. Norm. Super. Pisa Cl. Sci. (5)}, 14(1):155--220, 2015.

\bibitem[ASC07]{alon2007plA}
O.~E. Alon, A.~I. Streltsov, and L.~S. Cederbaum.
\newblock {Time-dependent multi-orbital mean-field for fragmented
  Bose–Einstein condensates}.
\newblock {\em Phys. Lett. A}, 362(5):453 -- 459, 2007.

\bibitem[ASC08]{alon2008prA}
O.~E. {Alon}, A.~I. {Streltsov}, and L.~S. {Cederbaum}.
\newblock {Multiconfigurational time-dependent Hartree method for bosons:
  Many-body dynamics of bosonic systems}.
\newblock {\em Phys. Rev. A}, 77(3):033613, 2008.

\bibitem[BdOS15]{benedikter2015cpam}
Niels Benedikter, Gustavo de~Oliveira, and Benjamin Schlein.
\newblock Quantitative derivation of the {G}ross-{P}itaevskii equation.
\newblock {\em Comm. Pure Appl. Math.}, 68(8):1399--1482, 2015.

\bibitem[BF09]{bader2009prl}
Philipp Bader and Uwe~R. Fischer.
\newblock {Fragmented Many-Body Ground States for Scalar Bosons in a Single
  Trap}.
\newblock {\em Phys. Rev. Lett.}, 103:060402, Aug 2009.

\bibitem[BPS16]{Benedikter-Porta-Schlein-2015}
Niels Benedikter, Marcello Porta, and Benjamin Schlein.
\newblock {\em {Effective evolution equations from quantum dynamics}}, volume~7
  of {\em Springer Briefs in Mathematical Physics}.
\newblock Springer, Cham, 2016.

\bibitem[BS19]{brennecke2017arxiv}
Christian Brennecke and Benjamin Schlein.
\newblock {Gross-Pitaevskii Dynamics for Bose-Einstein Condensates}.
\newblock {\em Anal. PDE}, 12(6):1513--1596, 2019.

\bibitem[BSS18]{Benedikter2018}
Niels Benedikter, J{\'e}r{\'e}my Sok, and Jan~Philip Solovej.
\newblock The dirac--frenkel principle for reduced density matrices, and the
  bogoliubov--de gennes equations.
\newblock {\em Annales Henri Poincar{\'e}}, 19(4):1167--1214, Apr 2018.

\bibitem[CF18]{correggi2017ahp}
Michele Correggi and Marco Falconi.
\newblock {Effective Potentials Generated by Field Interaction in the
  Quasi-Classical Limit}.
\newblock {\em Ann. Henri Poincar{\'e}}, 19(1):189--235, 2018.

\bibitem[CFO18]{correggi2017arxiv}
M.~Correggi, M.~Falconi, and M.~Olivieri.
\newblock {Magnetic Schr\"odinger Operators as the Quasi-Classical Limit of
  Pauli-Fierz-type Models}.
\newblock {\em J. Spectr. Theory}, 2018.
\newblock To appear.

\bibitem[ESY10]{ErdSchYau-10}
L{\'{a}}szl{\'{o}} Erd{\"{o}}s, Benjamin Schlein, and Horng-Tzer Yau.
\newblock {Derivation of the Gross-Pitaevskii equation for the dynamics of
  Bose-Einstein condensate}.
\newblock {\em Ann. of Math. (2)}, 172(1):291--370, 2010.

\bibitem[Fal13]{falconi2013jmp}
Marco Falconi.
\newblock Classical limit of the {N}elson model with cutoff.
\newblock {\em J. Math. Phys.}, 54(1):012303, 30, 2013.

\bibitem[Fal15]{falconi2015mpag}
Marco Falconi.
\newblock Self-adjointness criterion for operators in {F}ock spaces.
\newblock {\em Math. Phys. Anal. Geom.}, 18(1):Art. 2, 18, 2015.

\bibitem[Fal18a]{falconi2017ccm}
Marco Falconi.
\newblock {Concentration of cylindrical Wigner measures}.
\newblock {\em Commun. Contemp. Math.}, 20(5):1750055, 2018.

\bibitem[Fal18b]{falconi2016arxiv}
Marco Falconi.
\newblock {Cylindrical Wigner measures}.
\newblock {\em Doc. Math.}, 23:1677--1756, 2018.

\bibitem[FB10]{fischer2010pra}
Uwe~R. Fischer and Philipp Bader.
\newblock Interacting trapped bosons yield fragmented condensate states in low
  dimensions.
\newblock {\em Phys. Rev. A}, 82:013607, Jul 2010.

\bibitem[GV70]{ginibre1970cmpren}
J.~Ginibre and G.~Velo.
\newblock Renormalization of a quadratic interaction in the {H}amiltonian
  formalism.
\newblock {\em Comm. Math. Phys.}, 18:65--81, 1970.

\bibitem[GV79a]{ginibre1979cmpI}
J.~Ginibre and G.~Velo.
\newblock The classical field limit of scattering theory for nonrelativistic
  many-boson systems. {I}.
\newblock {\em Comm. Math. Phys.}, 66(1):37--76, 1979.

\bibitem[GV79b]{ginibre1979cmp2}
J.~Ginibre and G.~Velo.
\newblock The classical field limit of scattering theory for nonrelativistic
  many-boson systems. {II}.
\newblock {\em Comm. Math. Phys.}, 68(1):45--68, 1979.

\bibitem[HM76]{Hudson1976}
R.~L. Hudson and G.~R. Moody.
\newblock {Locally normal symmetric states and an analogue of de Finetti's
  theorem}.
\newblock {\em Zeitschrift f{\"{u}}r Wahrscheinlichkeitstheorie und Verwandte
  Gebiete}, 33(4):343--351, dec 1976.

\bibitem[KF14]{kang2014prl}
Myung-Kyun Kang and Uwe~R. Fischer.
\newblock {Revealing Single-Trap Condensate Fragmentation by Measuring
  Density-Density Correlations after Time of Flight}.
\newblock {\em Phys. Rev. Lett.}, 113:140404, Sep 2014.

\bibitem[KP10]{knowles2010cmp}
Antti Knowles and Peter Pickl.
\newblock Mean-field dynamics: singular potentials and rate of convergence.
\newblock {\em Comm. Math. Phys.}, 298(1):101--138, 2010.

\bibitem[Leg08]{leggett2008oup}
{Anthony James} Leggett.
\newblock {\em {Quantum Liquids: Bose condensation and Cooper pairing in
  condensed-matter systems}}.
\newblock Oxford University Press, 2008.

\bibitem[LNR14]{lewin2014am}
Mathieu Lewin, Phan~Th{{\`a}}nh Nam, and Nicolas Rougerie.
\newblock Derivation of {H}artree's theory for generic mean-field {B}ose
  systems.
\newblock {\em Adv. Math.}, 254:570--621, 2014.

\bibitem[LPB98]{law1998prl}
C.~K. Law, H.~Pu, and N.~P. Bigelow.
\newblock {Quantum Spins Mixing in Spinor Bose-Einstein Condensates}.
\newblock {\em Phys. Rev. Lett.}, 81:5257--5261, 1998.

\bibitem[Lub08]{lubich2008zlam}
Christian Lubich.
\newblock {\em From quantum to classical molecular dynamics: reduced models and
  numerical analysis}.
\newblock Zurich Lectures in Advanced Mathematics. European Mathematical
  Society (EMS), Z\"urich, 2008.

\bibitem[MHUB06]{Mueller2006}
Erich~J. Mueller, Tin-Lun Ho, Masahito Ueda, and Gordon Baym.
\newblock {Fragmentation of Bose-Einstein condensates}.
\newblock {\em Physical Review A}, 74(3):033612, sep 2006.

\bibitem[Pic10]{Pickl-10}
P~Pickl.
\newblock {Derivation of the time dependent {\{}G{\}}ross-{\{}P{\}}itaevskii
  equation without positivity condition on the interaction}.
\newblock {\em J. Stat. Phys.}, 140:76--89, 2010.

\bibitem[Pic15]{pickl2015rmp}
Peter Pickl.
\newblock Derivation of the time dependent {G}ross-{P}itaevskii equation with
  external fields.
\newblock {\em Rev. Math. Phys.}, 27(1):1550003, 45, 2015.

\bibitem[PO56]{penrose1956pr}
Oliver Penrose and Lars Onsager.
\newblock {Bose-Einstein Condensation and Liquid Helium}.
\newblock {\em Phys. Rev.}, 104:576--584, 1956.

\bibitem[RS09]{rodnianski2009cmp}
Igor Rodnianski and Benjamin Schlein.
\newblock Quantum fluctuations and rate of convergence towards mean field
  dynamics.
\newblock {\em Comm. Math. Phys.}, 291(1):31--61, 2009.

\bibitem[RS18]{Rougerie2016}
Nicolas Rougerie and Dominique Spehner.
\newblock {Interacting bosons in a double-well potential: localization regime}.
\newblock {\em Comm. Math. Phys.}, 361(2):737--786, 2018.

\bibitem[St{\o}69]{stormer1969jfa}
Erling St{\o}rmer.
\newblock Symmetric states of infinite tensor products of {C*}-algebras.
\newblock {\em J. Funct. Anal.}, 3:48--68, 1969.

\end{thebibliography}
\end{document}